\documentclass[pra,onecolumn,notitlepage,nofootinbib,floatfix,11pt,tightenlines]{revtex4-2}

\usepackage[utf8]{inputenc}
\usepackage{datetime}
\usepackage{graphics, graphicx, url, color, physics, cancel, multirow, bbm, soul, mathtools, mathrsfs, amsfonts, amssymb, amsmath, amsthm, dsfont, tabularx, bm, verbatim, tikz, caption, subcaption, xcolor,hyperref, nicefrac,mathdots}
\usepackage{bbm}
\usepackage{bm}
\usepackage{quantikz}
\usepackage[ruled,vlined]{algorithm2e}
\usepackage[boxsize=0.5em, centertableaux]{ytableau}
\usepackage{fullpage}
\newcommand{\RNum}[1]{\uppercase\expandafter{\romannumeral #1\relax}}

\theoremstyle{plain}

\newtheorem{theorem}{Theorem}
\newtheorem{corollary}[theorem]{Corollary}

\newtheorem{claim}{Claim}
\newtheorem{fact}[claim]{Fact}
\newtheorem{proposition}[theorem]{Proposition}

\theoremstyle{definition}
\newtheorem{definition}{Definition}
\newtheorem{example}{Example}

\theoremstyle{remark}

\newtheorem{remark}{Remark}

\theoremstyle{remark}

\def \vec {\normalfont{\text{vec}}}

\def \id {\mathbbm{1}}

\begin{document}
\title{Classification and implementation of unitary-equivariant and permutation-invariant quantum channels}

\author{Laura Mančinska}
\email{mancinska@math.ku.dk}
\affiliation{Centre for the Mathematics of Quantum Theory, University of Copenhagen}

\author{Elias Theil}
\email{edmt@math.ku.dk}
\affiliation{Centre for the Mathematics of Quantum Theory, University of Copenhagen}

\begin{abstract}
    Many quantum information tasks use inputs of the form $\rho^{\otimes m}$, which naturally induce permutation and unitary symmetries. We classify all quantum channels that respect both symmetries — i.e. unitary-equivariant and permutation-invariant quantum channels from $(\mathbb{C}^{d})^{\otimes m}$ to $(\mathbb{C}^{d})^{\otimes n}$ — via their extremal points. Operationally, each extremal quantum channel factors as \emph{unitary Schur sampling} $\rightarrow$ an \emph{irrep-level unitary-equivariant quantum channel} $\rightarrow$ the \emph{adjoint unitary Schur sampling}. We give a streaming implementation ansatz that uses an efficient streaming implementation of unitary Schur sampling together with a resource-state primitive, and we apply it to state symmetrization, symmetric cloning, and purity amplification. In these applications we obtain polynomial-time algorithms with exponential memory improvements in $m,n$. Further, for symmetric cloning we present, to our knowledge, the first efficient (polynomial-time) algorithm with explicit memory and gate bounds.
\end{abstract}

\maketitle

\tableofcontents

\section{Introduction}
A wide range of quantum tasks requires multiple copies of a given input state to perform some quantum operation. As an example, the quantum task of purity amplification \cite{Li_2025} takes $m$ copies of the noisy state
\begin{align}
    \rho_{\psi} = (1-\eta) \ketbra{\psi}{\psi} + \frac{\eta}{d}\id_{d}
\end{align}
and approximately outputs $n<m$ copies of the pure state, i.e. $\ketbra{\psi}{\psi}^{\otimes n}$.\footnote{For a more formal definition, see Section \ref{sec:examples}} To this end, the optimal algorithm makes use of two key symmetries present in this task. First, the input is given by $\rho_{\psi}^{\otimes m}$, so permuting the input registers doesn't change the outcome. In a similar way, permuting the registers for a given output doesn't change our measure of success, which is given by the fidelity with the ideal state $\ketbra{\psi}{\psi}^{\otimes n}$. We call this property \emph{permutation-invariance}. Second, acting on the input by some unitary as $(U\rho_{\psi} U^{*})^{\otimes m}$ leads to a rotated ideal output $(U\ketbra{\psi}{\psi}U^{*})^{\otimes n}$. We call this second property \emph{unitary-equivariance}.

Many tasks in quantum information theory are either unitary-equivariant like asymmetric cloning \cite{Nechita_2023}, (multi)port-based teleportation \cite{Ishizaka_2009,Kopszak_2021} or entanglement concentration \cite{Hayashi_2003}, or have both symmetries like spectrum estimation and related tasks \cite{ODonnell_2015}, state tomography \cite{ODonnell_2016,Haah_2017}, symmetric cloning \cite{Fan_2014}, (super)broadcasting \cite{Dang_2007,Chiribella_2007} and quantum majority vote \cite{Buhrman_2022}. The symmetries in these tasks translate directly to symmetries of the respective optimal protocols, which helps a lot with their analysis. We can understand the wide prevalence of these symmetries if we consider that any problem with input $\rho^{\otimes m}$ naturally exhibits the permutation symmetry, while any problem that is invariant under a local change of reference basis exhibits unitary symmetry.

\emph{Schur-Weyl duality} is a central tool for dealing with the symmetries listed above. It tells us that
\begin{align}
    (\mathbb{C}^{d})^{\otimes m} \stackrel{SU(d) , S_{m}}{\cong} \bigoplus_{\lambda \vdash_{d} m} \mathcal{P}_{\lambda}^{m} \otimes \mathcal{Q}_{\lambda}^{d} \, ,
\end{align}
where the $\mathcal{P}_{\lambda}^{m}$ and $\mathcal{Q}_{\lambda}^{d}$ are irreducible representations (irreps) of the groups $S_{m}$ and $SU(d)$ respectively, and the $\lambda \vdash_{d} m$ are the partitions that label them. The \emph{Schur transform}, which implements the above isometry, is part of most algorithms solving these problems. Its implementation on quantum computers, together with its extension to the \emph{mixed Schur transform}, were studied in the literature \cite{Bacon_2005,Krovi_2019,Grinko_2023,Nguyen_2023}.

There have been attempts to understand the shared structure and algorithmic tools for these problems beyond Schur-Weyl duality and the Schur transform. For permutation-invariant problems, recent work \cite{CerveroMartin_2024} has introduced the concept of \emph{unitary Schur sampling}. It uses the fact that a permutation-invariant input can be written via Schur-Weyl duality and Schur's lemma as
\begin{align}
    \label{equ:permutation_invariant_state}
    \rho \stackrel{SU(d) , S_{m}}{\cong} \bigoplus_{\lambda \vdash_{d} m} \frac{1}{\dim \mathcal{P}_{\lambda}^{m}} \id_{\mathcal{P}_{\lambda}^{m}} \otimes c_{\lambda} \rho_{\lambda} \, ,
\end{align}
where the $c_{\lambda}$ sum to $1$ and the $\rho_{\lambda}$ are states on $\mathcal{Q}_{\lambda}^{d}$. Since the register $\mathcal{P}_{\lambda}^{m}$ is in the maximally mixed state, it carries no information and can be discarded, which in turn leads to a general and efficient pre-processing algorithm that reduces memory requirements.

Unitary-equivariant quantum channels were discussed in \cite{Grinko_2023,Nguyen_2023,Nguyen_2024}, and general SDPs with unitary-equivariance were discussed in \cite{Grinko_2024}. These works make use of the so-called \emph{mixed Schur-Weyl duality} together with Schur's lemma to diagonalize the Choi matrix $C_{\Phi}$ of a given quantum channel $\Phi$ in a similar manner as Equation \eqref{equ:permutation_invariant_state}, which gives
\begin{align}
    C_{\Phi} \stackrel{SU(d) , \mathcal{A}_{m,n}^{d}}{\cong} \bigoplus_{\gamma\vdash_{d} (m,n)} X_{\gamma} \otimes \id_{\mathcal{Q}_{\gamma}^{d}} \, .
\end{align}
Here, the operators $X_{\gamma} \geq 0$ live on the irrep $\mathcal{P}_{\gamma}^{m,n,d}$ of the algebra of partially transposed permutations $\mathcal{A}_{m,n}^{d}$. For the case $\max(m,n) \ll d$, this leads to a large reduction in the number of SDP variables \cite{Grinko_2024}. Furthermore, an explicit implementation of quantum channels of this type for $n=1$ is given in \cite{Nguyen_2023} using the \emph{mixed Schur transform}.

We now turn towards quantum tasks with both permutation-invariance and unitary-equivariance. As discussed above, there is a wide range of such instances. We expect that the combined symmetry yields even more structure when it comes to classifying the optimal protocols and implementing them on a quantum computer. This motivates the study of general unitary-equivariant and permutation-invariant quantum channels, which have been discussed so far for the special cases $d=2$ and $n=1$ in \cite{Buhrman_2022} and for $\min(m,n)=1$ in \cite{Grinko_2024}.

In broader terms, equivariant quantum channels on irreps have been studied for $SU(2)$ in \cite{Nuwairan_2013,Aschieri_2024}, and for finte groups and special cases of irreps in \cite{Mozrzymas_2017}. Further, a full classification of group equivariant quantum operations\footnote{These are completely positive and trace non-increasing linear maps.} for continuous Lie groups has been given in \cite{Dariano_2004}. However, this result is very general and doesn't make use of the deeper understanding we have for the irreps of $SU(d)$ and $S_{m}$. Additionally, all classifications of equivariant quantum channels described above give a description in terms of the Choi matrix, without providing an understanding of how to implement them on a quantum computer.

We now list our main contributions:
\newline 

\noindent
\textbf{Contribution I: Classify and interpret.} (Theorem \ref{thm:result_1}) We classify all unitary-equivariant and permutation-invariant quantum channels from $(\mathbb{C}^{d})^{\otimes m}$ to $(\mathbb{C}^{d})^{\otimes n}$ in the Choi matrix formalism, via the set of all extremal quantum channels. In addition, we show that they factor as combinations of unitary Schur sampling \cite{CerveroMartin_2024}, unitary-equivariant quantum channels on irreps (studied e.g. in \cite{Aschieri_2024}), and the dual quantum channel of unitary Schur sampling, as depicted in Figure \ref{fig:factoring_circuit_diagram_start}
\newline

\begin{figure}[h!]
    \centering
    \scalebox{1}{\input{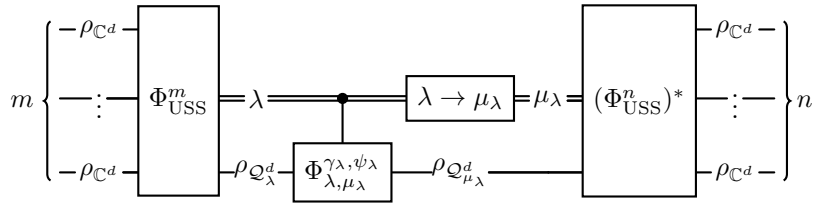}}
    
    \caption{Circuit diagram showing the factoring of quantum channels in Theorem \ref{thm:operational_interpretation_extremal_quantum channels}. Unitary Schur sampling $\Phi_{\rm USS}^{m}$ provides an output irrep label $\lambda$, together with a state on that output irrep $\rho_{\mathcal{Q}_{\lambda}^{d}}$. For each label $\lambda$, a unitary-equivariant quantum channel $\Phi_{\lambda,\mu_{\lambda}}^{\gamma_{\lambda},\psi_{\lambda}}$ maps the unitary irrep $\mathcal{Q}_{\lambda}^{d}$ into the unitary irrep $\mathcal{Q}_{\mu_{\lambda}}^{d}$, which is then transformed back via dual unitary Schur sampling. The superscripts $\gamma_{\lambda},\psi_{\lambda}$ denote a particular choice of quantum channel, explained in detail in Definition \ref{def:unitary_equivariant_quantum channel_irreps}.}
    \label{fig:factoring_circuit_diagram_start}
\end{figure}

While the mathematical description of optimal quantum channels for problems such as purity amplification and symmetric cloning have been studied well, their implementation as quantum circuits have received less attention. In \cite{Li_2025}, there is an efficient implementation, which is specific to purity amplification. For the case of more general unitary-equivariant and permutation-invariant quantum channels, only $d=2$ and $n=1$ has been discussed in \cite{Buhrman_2022}. In particular, to our knowledge, there exists no efficient known algorithm to implement symmetric cloning, apart from approaches with a specialized experimental setup \cite{Fan_2014}, or for particular parameter values like $d=2$ and $m=1$ \cite{Pelofske_2022}. This motivates our second and third contributions:
\newline

\noindent
\textbf{Contribution II: General ansatz.} (Theorem \ref{thm:result_2}) We provide a general ansatz to implement any unitary-equivariant and permutation-invariant quantum channel from $(\mathbb{C}^{d})^{\otimes m}$ to $(\mathbb{C}^{d})^{\otimes n}$ in a streaming manner. This ansatz reduces the task of implementing such a quantum channel in general to the implementation of unitary Schur sampling and its dual quantum channel, together with the implementation of a unitary-equivariant quantum channel on a unitary irrep via a special resource state. This resource state is in general hard to obtain. However, as we show with our examples, in many relevant cases it has a simple form and can be initialized in polynomial time.
\newline

\noindent
\textbf{Contribution III: Three example applications.} (Theorem \ref{thm:result_3}) 
We apply the general ansatz to the quantum tasks of purity amplification, symmetric cloning and state symmetrization. We introduce the task of state symmetrization, where a random permutation is applied to a given state. In all three cases, we obtain an efficient streaming algorithm with exponential improvement in $(m,n)$ for memory complexity over existing algorithms. For the case of state symmetrization, this result is surprising, as we would expect the need to store all $m$ registers to apply a random permutation. For the case of purity amplification, our algorithm outperforms the known algorithm in terms of memory and gate complexity in the regime $d \ll m$. For the case of symmetric cloning, we provide, to our knowledge, the first efficient algorithm.
\newline

This paper is structured as follows: First, we state the results (Section \ref{sec:technical_results}) and discuss them (Section \ref{sec:discussion}). Afterwards, we introduce notation and prelminaries (Section \ref{sec:mathematical_preliminaries_notation}). Then, we classify the unitary-equivariant and permutation-invariant quantum channels (Section \ref{sec:classification_and_operational_interpretation}). Finally, we give a general ansatz for implementation (Section \ref{sec:general_ansatz}) and apply this ansatz to three examples (Section \ref{sec:examples}).

\section{Technical results}
\label{sec:technical_results}

Let $\mathcal{C}^{d}(m,n)$ denote the set of completely positive and trace preserving linear maps, i.e. quantum channels, from $(\mathbb{C}^{d})^{\otimes m}$ to $(\mathbb{C}^{d})^{\otimes n}$. We define the main object of interest for this work. 

\begin{definition}
    A quantum channel $\Phi\in \mathcal{C}^{d}(m,n)$ is \textit{unitary-equivariant} if for all $U\in SU(d)$ we have
    \begin{align}
        \Phi(U^{\otimes m}A(U^{*})^{\otimes m})=U^{\otimes n}\Phi(A)(U^{*})^{\otimes n} \, .
    \end{align}
    Similarly, $\Phi$ is \textit{permutation-invariant} if for all $\sigma\in S_{m}$ and $\tau \in S_{n}$ we have
    \begin{align}
        \Phi(\sigma A \sigma^{*})=\Phi(A)=\tau\Phi(A)\tau^{*} \, .
    \end{align}
    The set of unitary-equivariant and permutation-invariant quantum channels is $\mathcal{C}_{up}^{d}(m,n) \subseteq \mathcal{C}^{d}(m,n)$.
\end{definition}

\subsection{Classification and operational interpretation}

We need notation from representation theory to state our first result. An overview of the general notation and mathematical background is given in Section \ref{sec:mathematical_preliminaries_notation}. For a partition $\lambda \vdash_{d} m$, let $\overline{\lambda}$ be the partition labelling the dual irrep of $\mathcal{Q}_{\lambda}^{d}$. Let further $\lambda +_{d} \mu$ denote the set of partitions $\gamma$ so that $\mathcal{Q}_{\gamma}^{d}$ is contained in the irrep decomposition of $\mathcal{Q}_{\lambda}^{d} \otimes \mathcal{Q}_{\mu}^{d}$, and let $c_{\lambda,\mu}^{\gamma}$ be the respective multiplicity for a given $\gamma \in \lambda +_{d} \mu$. We now state a full classification of the extremal points of $\mathcal{C}_{up}^{d}(m,n)$, which is equivalent to a full classification of $\mathcal{C}_{up}^{d}(m,n)$ by taking convex combinations. In addition, we give an operational interpretation in terms of unitary Schur sampling and extremal quantum channels on $SU(d)$ irreps.

\begin{theorem}[Corollary \ref{cor:classification_extremal_points_Cusd} and Theorem \ref{thm:operational_interpretation_extremal_quantum channels}]
    \label{thm:result_1}
    Each extremal point $\Phi \in \mathcal{C}_{up}^{d}(m,n)$ corresponds to a collection of touples $(\mu_{\lambda},\gamma_{\lambda},\ket{\psi_{\lambda}})$ for each $\lambda\vdash_{d} m$, so that
    \begin{align}
        \mu_{\lambda} \vdash_{d} n \quad , \quad \gamma_{\lambda} \in \overline{\lambda} +_{d} \mu_{\lambda} \quad , \quad \ket{\psi_{\lambda}} \in \mathbb{C}^{c_{\overline{\lambda},\mu_{\lambda}}^{\gamma_{\lambda}}} \, .
    \end{align}
    For a given collection of touples, the Choi matrix $C_{\Phi}$ is
    \begin{align}
        C_{\Phi} \stackrel{SU(d),S_{m},S_{n}}{\cong} \bigoplus_{\lambda\vdash_{d} m}\id_{\mathcal{P}_{\lambda}^{m}}\otimes \frac{1}{\dim \mathcal{P}_{\mu_{\lambda}}^{n}}\id_{\mathcal{P}_{\mu_{\lambda}}^{n}} \otimes \frac{\dim \mathcal{Q}_{\lambda}^{d}}{\dim \mathcal{Q}_{\gamma_{\lambda}}^{d}}\id_{\mathcal{Q}_{\gamma_{\lambda}}^{d}}\otimes \ketbra{\psi_{\lambda}}{\psi_{\lambda}} \, .
    \end{align}
    Such a quantum channel $\Phi$ is equivalent to
    \begin{align}
        \label{equ:factoring_quantum channels}
        \Phi=(\Phi_{\rm USS}^{n})^{*}\circ \left(\bigoplus_{\lambda\vdash_{d}m}\Phi_{\lambda,\mu_{\lambda}}^{\gamma_{\lambda},\psi_{\lambda}}\right) \circ \Phi_{\rm USS}^{m} \, .
    \end{align}
    Here, $\Phi_{\rm USS}^{m}$ is the unitary Schur sampling quantum channel on $(\mathbb{C}^{d})^{\otimes m}$ described in \cite{CerveroMartin_2024}, $(\Phi_{\rm USS}^{m})^{*}$ is its adjoint with respect to the trace inner product, and the $\Phi_{\lambda,\mu}^{\gamma,\psi}$ are extremal quantum channels in the set of unitary-equivariant quantum channels from $\mathcal{Q}_{\lambda}^{d}$ to $\mathcal{Q}_{\mu}^{d}$, see Definition \ref{def:unitary_equivariant_quantum channel_irreps} and Theorem \ref{thm:equivalence_of_unitary_equivariant_quantum channels}.
\end{theorem}

The first part of Theorem \ref{thm:result_1} can be useful when optimizing over $\mathcal{C}_{up}^{d}(m,n)$ to solve unitary-equivariant and permutation-invariant problems, while the second part gives the operational interpretation.

\subsection{General ansatz}

For the second result of this work, we give an ansatz to implement the quantum channels described in Theorem \ref{thm:result_1}.

\begin{theorem}[Propositions \ref{prop:paths_embedding} and \ref{prop:implement_paths_embedding}, and Theorem \ref{thm:complexity_unitary_equivariant_permutation_invariant_quantum channels}]
    \label{thm:result_2}
    Let $\Phi\in\mathcal{C}_{up}^{d}(m,n)$ be an extremal quantum channel denoted by the touples $(\mu_{\lambda},\gamma_{\lambda},\ket{\psi_{\lambda}})$ as described in Theorem \ref{thm:result_1}, let $\sigma$ be a state on $S^{\otimes m}\subseteq (\mathbb{C}^{d})^{\otimes m}$, with $\dim S=r$, and let $l(\mu_{\lambda})\leq r'$ for all $\lambda$ with $l(\lambda)\leq r$. Then we implement the operation $\sigma\rightarrow\Phi(\sigma)$ in a streaming manner up to error $\epsilon$ in diamond norm with memory complexity $M$ and gate complexity $T$, where
    \begin{align}
        &M=O\big((r+r')d\log_{2}^{p}(d,m,n,1/\epsilon)\big) + \max\limits_{\lambda\vdash_{d} m} M_{\lambda,\mu_{\lambda}}^{\gamma_{\lambda},\psi_{\lambda}} \, ,\\
        \label{equ:gate_complexity_reduced}
        &T=O\big((mr^{3}+n(r')^{3})d\log_{2}^{p}(d,m,n,1/\epsilon)\big) + \max\limits_{\lambda\vdash_{d} m} T_{\lambda,\mu_{\lambda}}^{\gamma_{\lambda},\psi_{\lambda}} \, .
    \end{align}
    Here, $p\approx 1.44$, and $M_{\lambda,\mu}^{\gamma,\psi}$ and $T_{\lambda,\mu}^{\gamma,\psi}$ denote the memory and y necessary to implement the quantum channel $\Phi_{\lambda,\mu}^{\gamma,\psi}$. In addition, if given access to a problem-specific resource state $\ket{p_{\mu\rightarrow \lambda}^{k,l,d}}$, we have
    \begin{align}
        &M_{\lambda,\mu}^{\gamma,\psi} = O(d(\tilde{r}+k+l)\log_{2}^{p}(d,m,n,k,l,1/\epsilon)) \, ,\\
        &T_{\lambda,\mu}^{\gamma,\psi} = O\big((k+l)d\tilde{r}^{3} \log_{2}^{p}(d,m,n,k,l,1/\epsilon)\big) \, ,
    \end{align}
    The resource state $\ket{p_{\mu\rightarrow\lambda}^{k,l,d}}$ lives on $k+l+1$ registers, each encoding a staircase $\nu^{i}$, and $\tilde{r} = \max\limits_{0\leq i \leq k+l}l(\nu^{i})$, see Section \ref{subs:iterated_simple_CG_transforms}.
\end{theorem}

This theorem uses the operational interpretation given in Theorem \ref{thm:result_1} to provide an algorithm for a general $\Phi \in \mathcal{C}_{up}^{d}(m,n)$. The efficient streaming implementation of $\Phi_{\rm USS}^{m}$ has been discussed in \cite{CerveroMartin_2024}, and the respective algorithm can be adapted to $(\Phi_{\rm USS}^{n})^{*}$, see Proposition \ref{prop:symmetrization_quantum channel_implementation}. The main difficulty is the implementation of the quantum channels $\Phi_{\lambda,\mu}^{\gamma,\psi}$. An efficient implementation of general Clebsch--Gordan transforms $U_{\rm CG}^{\lambda,\mu,d}$ would solve this problem, see Theorems \ref{thm:complexity_unitary_equivariant_permutation_invariant_quantum channels} and \ref{thm:equivalence_of_unitary_equivariant_quantum channels}. However we expect that this is a hard problem, on not feasible for general $\lambda,\mu$. Therefore we reduce the problem to implementing a sequence of simple (dual) Clebsch--Gordan transforms $U_{\rm CG}^{\lambda, \ydiagram{1},d}$ and $U_{\rm CG}^{\lambda, \overline{\ydiagram{1}},d}$, see Propositions \ref{prop:paths_embedding} and \ref{prop:implement_paths_embedding}. We then transfer the complexity to the resource state input, given by superpositions over paths in the Bratteli diagram (see \cite{Grinko_2023,Nguyen_2023}) from $\lambda$ to $\mu$ or to $\gamma$. A detailed exposition of this approach is given in Section \ref{subs:iterated_simple_CG_transforms}.

Again, preparing the correct resource state for general quantum channels is likely hard, as it corresponds to finding arbitrary $\mathcal{A}_{k,l}^{d}$ irreps inside $\mathcal{A}_{m+k,n+l}^{d}$ irreps. However, the quantum channels that appear as solutions of unitary-equivariant and permutation-invariant problems often have more structure in the respective touples $(\mu_{\lambda},\gamma_{\lambda},\ket{\psi_{\lambda}})$. For example in the cases of purity amplification, symmetric cloning and state symmetrization, all the relevant multiplicity spaces are one-dimensional, and therefore $\ket{\psi_{\lambda}}$ is always trivial. Things also become easier if one or more of the involved irreps are symmetric irreps (symmetric cloning, purity amplification), the trivial irrep (state symmetrization), or if one label can be written as the sum of other labels (symmetric cloning has $\lambda + \gamma_{\lambda} = \mu_{\lambda}$).

The implementation of Theorem \ref{thm:result_2} uses mid-circuit measurements coming from \cite{CerveroMartin_2024}, and coherent floating-point arithmetic for the Clebsch-Gordan transforms \cite{CerveroMartin_2024,Bacon_2005,Grinko_2023,Nguyen_2023}. These requirements are challenging for current hardware and likely need good error-correction to be implemented.

\subsection{Application to three tasks}

Our third result gives applied instances of the ansatz.

\begin{theorem}[Corollaries \ref{cor:result_state_symmetrization}, \ref{cor:result_symmetric_cloning} and \ref{cor:result_purity_amplification}]
    \label{thm:result_3}
    We implement the following quantum channels in a streaming manner up to error $\epsilon$ in diamond norm:
    \begin{enumerate}
        \item The symmetrization quantum channel $\rho\rightarrow\Phi_{\rm sym}^{m}(\rho)$ for an input state $\rho$ on $S^{\otimes m}\subseteq (\mathbb{C}^{d})^{\otimes m}$ with $\dim S = r$, which has memory and gate complexity as given in Table \ref{tab:comparison_algorithms_symmetrization}.
        \item The symmetric cloning quantum channel $\ketbra{\psi}{\psi}^{\otimes m} \rightarrow \Phi_{\rm cl}^{m,n}(\ketbra{\psi}{\psi}^{\otimes m})$, which has memory and gate complexity as given in Table \ref{tab:comparison_algorithms_symmetric_cloning}.
        \item The purity amplification quantum channel $\rho \rightarrow \Phi_{\rm pa}^{m,1}(\rho)$, which has memory and gate complexity as given in Table \ref{tab:comparison_algorithms_purity_amplification}.
    \end{enumerate}
    A detailed description of these quantum channels is given in Section \ref{sec:examples}.
\end{theorem}

\begingroup
\begin{table*}[htbp!]
\renewcommand{\arraystretch}{1.5}
\begin{center}
    \begin{tabular}{ p{0.30\linewidth} p{0.35\linewidth} p{0.35\linewidth} } 
    \hline
        Algorithm & Memory & Gates \\ 
    \hline
        Random permutation & $O\big(m\log_{2}(d)\big)$ & $O\big(m\big)$ \\
    \hline
        This work & $O\big(rd\log_{2}^{p}(d,m,1/\epsilon)\big)$ & $O\big(mr^{3}d\log_{2}^{p}(d,m,1/\epsilon)\big)$ \\
    \hline
    \end{tabular}
\end{center}
\caption{Complexity for state symmetrization.
}
\label{tab:comparison_algorithms_symmetrization}
\end{table*}

\begin{table*}[htbp!]
\renewcommand{\arraystretch}{1.5}
\begin{center}
    \begin{tabular}{ p{0.23\linewidth} p{0.25\linewidth} p{0.32\linewidth} p{0.2\linewidth} } 
    \hline
        Algorithm & Memory & Gates & Remarks \\
    \hline
        NISQ Cloning \cite{Pelofske_2022} & $O\big(n\log_{2}^{p}(n,1/\epsilon)\big)$ & $O\big( n\big)$ & $m=1$, $d=2$ \\
    \hline
        Choi matrix state \cite{Nguyen_2023} & $O\big((n+d)\log_{2}^{p}(d,n,1/\epsilon)\big)$ & $O\big( nd^{4} \log_{2}^{p}(d,n,1/\epsilon) + (d^{6}+n)\big)$ & $m=O(1)$ \\
    \hline
        Direct implementation & $O\big(n\log_{2}(d)\big)$ & $O\big(((n+d)/k)^{2k} \log_{2}^{p}(1/\epsilon)\big)$ & Stinespring dilation and Solovay-Kitaev construction; $k=\min(n,d)$ \\
    \hline
        This work & $O\big(d\log_{2}^{p}(d,n,1/\epsilon)\big)$ & $O\big(nd\log_{2}^{p}(d,n,1/\epsilon)\big)$ & \\
    \hline
    \end{tabular}
\end{center}
\caption{Complexity for symmetric cloning.}
\label{tab:comparison_algorithms_symmetric_cloning}
\end{table*}

\begin{table*}[htbp!]
\renewcommand{\arraystretch}{1.5}
\begin{center}
    \begin{tabular}{  p{0.30\linewidth} p{0.35\linewidth} p{0.35\linewidth} } 
    \hline
        Algorithm & Memory & Gates \\ 
    \hline
        Quantum Fourier transform \cite{Li_2025} & $O\big(m\log_{2}(m)\big)$ & $O\big(\text{poly} (m,\log_{2}(d,1/\epsilon))\big)$ \\
    \hline
        This work & $O\big(d^{2}\log_{2}^{p}(d,m,1/\epsilon)\big)$ & $O\big(md^{4}\log_{2}^{p}(d,m,1/\epsilon)\big)$ \\
    \hline
    \end{tabular}
\end{center}
\caption{Complexity for purity amplification.
}
\label{tab:comparison_algorithms_purity_amplification}
\end{table*}
\endgroup

The result for state symmetrization is surprising, as we only need $\log(m)$ memory to apply a random permutation on a state on $m$ systems. We compare our algorithm for state symmetrization to the straightforward implementation via application of a random permutation.

For symmetric cloning, our result gives, to our best knowledge, the first implementation and analysis of memory and gate complexity for the problem of symmetric cloning with local dimension $d$. There are other approaches described in \cite{Fan_2014}, but they are either specific to certain physical implementations or without explicit analysis of memory and gate complexity. We compare our algorithm to various algorithms given in the literature and to a direct implementation via Stinespring dilation and an improved Solovay-Kitaev construction. The direct implementation has an output space of $(\mathbb{C}^{d})^{\otimes n}$, which gives a memory complexity of $O(n\log_{2}(d))$. In addition, the implementation via Stinespring dilation and a Solovay-Kitaev construction has $T=O(D^{2}\log_{2}(1/\epsilon))$, where $D$ is the maximum of the input and output subspace dimensions. In our case we take $D$ to be the dimension of the symmetric subspace, and we get
\begin{align}
    D = d[n] =
    \begin{pmatrix}
        n + d - 1 \\
        n
    \end{pmatrix}
    \geq ((n+d)/\min(n,d))^{\min(n,d)} \, .
\end{align}

Finally, our result for purity amplification gives an alternative to the algorithm presented in \cite{Li_2025} for the regime $d \ll m$, while their algorithm is advantageous in the regime $m \ll d$.

\section{Discussion}
\label{sec:discussion}
We believe the results presented in this work to be useful in multiple ways.
\newline 

\noindent
\textbf{Classification.} Our classification is a valuable mathematical tool when tackling new problems with unitary and permutation symmetries, such as broadcasting problems, purity amplification with $n>1$ or majority vote with local dimension $d>2$.
\newline 

\noindent
\textbf{General ansatz.} We expect that our general ansatz will provide an ansatz for the implementation of the quantum channels that optimize the problems listed above.
\newline 

\noindent
\textbf{Three example applications.} We provide three new explicit algorithms that are polynomial time in $n,m,d$ for the tasks of state symmetrization, symmetric cloning and purity amplification. We reiterate that our algorithm for symmetric cloning is, to our knowledge, the first efficient algorithm for general $d$. In the case of purity amplification, we provide an alternative algorithm for the regime $d\ll m,n$, and in the case of state symmetrization we obtain a surprising exponential advantage in memory for the same regime.
\newline

Future work includes the applications of broadcasting, purity amplification and majority vote, but also other directions.
\newline 

\noindent
\textbf{General CG transforms.} The quantum channels $\Phi_{\lambda,\mu}^{\gamma,\psi}$ can be implemented in particular, if we can implement general Clebsch-Gordan transforms. Maybe there is actually a way to implement them efficiently on a quantum computer.
\newline 

\noindent
\textbf{The case $m,n \ll d$.} There have been interesting approaches to the Schur transform and quantum channels in $\mathcal{C}_{up}^{d}(m,n)$ for the case $m \ll d$ \cite{Krovi_2019,Li_2025}, upcoming work of Grinko et al. The idea is to encode the alphabet of $d$ symbols into a smaller alphabet of at most $m$ symbols, which gives memory and time complexity of $\text{poly}(\log_{2}(d))$ in $d$. This approach works well for the simple CG transform, but potentially runs into problems with the simple dual CG transform. It is interesting how far the results of this paper then are translatable into the setting $m,n \ll d$. We know that purity amplification works \cite{Li_2025}. We suspect the same from state symmetrization, and in general any problem where the path described in Section \ref{sec:general_ansatz} only uses the simple CG transform. Whether it can be extended to a more general setting, which would include symmetric cloning, remains open.

\section{Mathematical preliminaries and notation}
\label{sec:mathematical_preliminaries_notation}

We first review some general notation and mathematical preliminaries.
\newline

\noindent \emph{General. } We use \emph{iff} instead of the expression \emph{if and only if}. 
\newline

\noindent \emph{Algorithms. } We use $\log_{q}(x,y,z)$ to denote an expression of the form $\log_{q}(\text{poly}(x,y,z))$. Further, by \emph{gate complexity} we understand the number of elementary ($CNOT$, $H$, $Z^{\nicefrac{1}{4}}$) qubit gates, and by \emph{memory complexity} we understand the total number of qubit registers necessary to perform the calculation. Finally, by \emph{streaming}, we mean that an algorithm can access the input registers sequentially, and it can output the output registers sequentially.

\subsection{Linear algebra and quantum channels}
We discuss our operator notation and the Choi formalism.
\newline

\noindent \emph{Bounded linear operators. } For vector space $V_{1},V_{2}$, let $\mathcal{B}(V_{1},V_{2})$ be the set of bounded linear operators from $V_{1}$ to $V_{2}$. If $V_{1}=V_{2}=V$, we write $\mathcal{B}(V)$ instead. For $A\in \mathcal{B}(V)$ we denote by $\overline{A}$ the complex conjugate of $A$, by $A^{T}$ the transpose and by $A^{*} = \overline{A}^{T}$ the adjoint. Both complex conjugate and transpose depend on the choice of basis, and we indicate this choice where relevant. In addition, we denote the identity on $V$ as $\id_{V}$. Let finally $W_{1} \subseteq V_{1}$ be a subspace of $V_{1}$ and $A\in \mathcal{B}(V_{1},V_{2})$. Then we denote by $A|_{W_{1}}\in \mathcal{B}(W_{1},V_{2})$ the restriction of $A$ to $W_{1}$.
\newline

\noindent \emph{Direct sums and tensor products. } Given two vector space $V_{1},V_{2}$, we denote their direct sum by $V_{1}\oplus V_{2}$ and their tensor product by $V_{1} \otimes V_{2}$. For $m$ tensor factors of the same vector space we write $V^{\otimes m}$. We use the same notation for operators on the respective direct sums or tensor products. For $V=V_{1} \oplus V_{2}$ and $A\in\mathcal{B}(V_{1})$, we interpret $A\in\mathcal{B}(V)$ as $A \oplus 0$. In particular, we omit all zero terms in direct sums.
\newline

\noindent \emph{quantum channels. } We denote by $\mathcal{C}(V_{1},V_{2}) \subseteq \mathcal{B}(\mathcal{B}(V_{1}),\mathcal{B}(V_{2}))$ the set of all completely positive and trace preserving linear maps, i.e. quantum channels. If $V_{1}=(\mathbb{C}^{d})^{\otimes m}$ and $V_{2}=(\mathbb{C}^{d})^{\otimes n}$, we write $\mathcal{C}^{d}(m,n)$ instead. For $\Phi \in \mathcal{C}(V_{1},V_{2})$, we denote by $\Phi^{*} \in \mathcal{B}(\mathcal{B}(V_{2}),\mathcal{B}(V_{1}))$ the adjoint of $\Phi$ with respect to the trace inner product. Finally, we denote the Choi matrix of $\Phi$ by $C_{\Phi}$, which is given as
\begin{align}
    C_{\Phi} := \sum_{i=1}^{\dim V_{1}} \sum_{j=1}^{\dim V_{1}} \ketbra{i}{j} \otimes \Phi(\ketbra{i}{j}) \, .
\end{align}

\subsection{Representations and dual representations}

We assume familiarity with the notions of group and algebra representations and irreducible representations (irreps).\footnote{In the context of algebras, representations are usually called \emph{modules}, and irreps are called \emph{simple modules}. To unify notation, we use the words representation and irrep for both group and algebra representations.}
\newline

\noindent \emph{Dual representations. } We take a group representation $V$ of $G$ with the action $g \rightarrow r(g)$. Then the \emph{dual representation} is given by
\begin{align}
    g \rightarrow r(g^{-1})^{T} \, .
\end{align}
For unitary representations, we have $r(g^{-1})^{T} = \overline{r(g)}$, and we denote the vector space of this representation by $\overline{V}$.
\newline

\noindent \emph{Isomorphic representations. } Let $V_{1},V_{2}$ be representations of the group $G$ with corresponding actions $r_{1},r_{2}$. Let further $A\in \mathcal{B}(V_{1},V_{2})$ so that
\begin{align}
    \label{equ:operator_commutes_group_action}
    A r_{1}(g) = r_{2}(g) A \,
\end{align}
for all $g\in G$. Then we say that $A$ \emph{commutes} with the action of $G$.\footnote{Operators that fulfil Equation \eqref{equ:operator_commutes_group_action} are often called \emph{intertwiners}.} If there exists an isomorphism between $V_{1}$ and $V_{2}$ that commutes with the action of $G$, then we call the representations \emph{isomorphic}, and we write $V_{1} \stackrel{G}{\cong} V_{2}$. If the isomorphism commutes with the action of two groups $G_{1},G_{2}$ at the same time, we write $V_{1} \stackrel{G_{1},G_{2}}{\cong} V_{2}$.\footnote{If $r_{1},r_{2}$ are unitary representations, we can always choose this isomorphism to be an isometry.} We use a similar notation for algebra representations, and when an isomorphism commutes with both a group and an algebra action.
\newline

\subsection{Partitions and staircases}

Staircases (generalizations of partitions) index the irreps of the special unitary group $SU(d)$ and its commutant, the algebra partially transposed permutations $\mathcal{A}_{m,n}^{d}$. We give a general overview of them.
\newline

\noindent \emph{Staircases. } We define a \emph{staircase} $\gamma \vdash_{d} (m,n)$ to be a tuple of integers $\gamma \in \mathbb{Z}^{d}$ so that
\begin{align}
    \gamma_{1} \geq \gamma_{2} \geq ... \geq \gamma_{d} \, 
\end{align}
and
\begin{align}
    \sum_{i:\gamma_{i}\geq 0} \gamma_{i} \leq m \quad , \quad \sum_{i:\gamma_{i}\leq 0} \gamma_{i} \geq -n \quad , \quad \sum_{i} \gamma_{i} = m - n \, .
\end{align}
In the case of $n=0$, this leads to $\gamma_{d} \geq 0$. Such a staircase is called a \emph{partition of $m$}, and we denote this by $\gamma \vdash_{d} m$. For a staircase $\gamma$, we call the number of nonzero entries in $\gamma$ the \emph{length} of $\gamma$, and we denote it by $l(\gamma)$.
\newline

\noindent \emph{Indexing of staircases. } For a given staircase $\gamma \vdash_{d} (m,n)$, the entries $\gamma_{i}$ are marked by subscripts. When we want to index a family of staircases, we use superscripts, e.g. $\gamma^{1},\gamma^{2},...$, where each $\gamma^{i}$ is itself a staircase.
\newline

\noindent \emph{Skew Young diagrams. } A common way to visualize staircases is via \emph{skew Young diagram}. Each skew Young diagram is a collection of rows of boxes stacked on top of each other, starting with $\gamma_{1}$ at the top. The $i$-th row contains $|\gamma_{i}|$ boxes. If $\gamma_{i} \geq 0$, then the boxes extend to the right, and if $\gamma_{i} \leq 0$ then they extend to the left. When $n=0$, we have boxes only extending to the right, and we call the resulting diagrams simply \emph{Young diagrams}. A special case is given by $\gamma \vdash_{d} 1$, where the corresponding Young diagram is given by a single box $\ydiagram{1}$. Throughout this work, we use staircases and their respective skew Young diagrams interchangeably.

\begin{example}
    Let $\gamma \vdash_{4} (4,2)$ be given by $\gamma = (2,1,0,-1)$. Then the skew Young diagram corresponding to $\gamma$ is given by
    \begin{align}
        \ytableausetup{boxsize=1em}
        \gamma \sim \ydiagram{1+2,1+1,1+0,1} \, .
        \ytableausetup{boxsize=0.5em}
    \end{align}
\end{example}

\subsection{Representation theory of $SU(d)$}

One of the two central symmetries encountered in this work is given by the special unitary group on $\mathbb{C}^{d}$, which we label $SU(d)$. We recall the general irrep labels and Clebsch–Gordan rules—the backbone for coupling $SU(d)$ irreps.
\newline

\noindent \emph{Irreps of $SU(d)$. } The irreducible representations of $SU(d)$ are labeled by staircases $\gamma \vdash_{d} (m,n)$ for any $m,n \geq 0$. We denote the corresponding vector space by $\mathcal{Q}_{\gamma}^{d}$, and we fix the basis and action of $SU(d)$ to be the same as in \cite{Grinko_2023,Nguyen_2023,Vilenkin_1995}. Two irreps given by $\gamma^{1} \vdash_{d} m_{1}$ and $\gamma^{2} \vdash_{d} m_{2}$ are isomorphic iff we have that $\gamma^{1} = \gamma^{2} + (k,...,k)$ for some $k\in \mathbb{Z}$.\footnote{By continuity, we extend $\mathcal{Q}_{\gamma}^{d}$ to be a representation of the unitary group $U(d)$ and even the general linear group $GL(d)$. However, then the addition of entries $(k,...,k)$ corresponds to multiplying the representation by $\det^{k}$, and therefore irreps with different labels are no longer isomorphic.} This means in particular that we always find a partition $\lambda$ with only positive entries to describe a given irrep of $SU(d)$. An important irrep is the defining representation $\mathbb{C}^{d}$. Its label is given by $(1,0,...,0) \sim \ydiagram{1}$. For a given staircase $\gamma = (\gamma_{1},...,\gamma_{d})$, the dual irrep is labeled by $\overline{\gamma} = (-\gamma_{d},...,-\gamma_{1})$. For example, the dual defining representation $\overline{\mathbb{C}^{d}}$ is labeled by the staircase $\overline{(1)} = (0,...,0,-1) \sim \overline{\ydiagram{1}}$. We interchangeably write $\overline{\mathcal{Q}_{\gamma}^{d}} = \mathcal{Q}_{\overline{\gamma}}^{d}$. For example, the dual defining representation $\overline{\mathbb{C}^{d}}$ is labeled by the staircase $\overline{(1)} = (0,...,0,-1) \sim \overline{\ydiagram{1}}$.
\newline

\noindent \emph{Clebsch--Gordan transforms. } For two $SU(d)$ irreps $\mathcal{Q}_{\lambda}^{d}$ and $\mathcal{Q}_{\mu}^{d}$, we decompose their tensor product again into irreps
\begin{align}
    \label{equ:unitary_irreps_decomposition}
    \mathcal{Q}_{\lambda}^{d} \otimes \mathcal{Q}_{\mu}^{d} \stackrel{SU(d)}{\cong} \bigoplus_{\gamma \in \lambda +_{d} \mu} \mathcal{Q}_{\gamma}^{d} \otimes \mathbb{C}^{c_{\lambda,\mu}^{\gamma}} \, .
\end{align}
Here, $c_{\lambda,\mu}^{\gamma}$ is often called the \emph{Littlewood-Richardson coefficient}, and we denote by $\lambda +_{d} \mu$ the set of all $\gamma$ so that $c_{\lambda,\mu}^{\gamma} \geq 1$. An isometry $U_{\rm CG}^{\lambda,\mu,d}$ which achieves the isomorphism above is called a \emph{Clebsch-Gordan transform}, and it is unique up to a unitary rotation on the multiplicity spaces $\mathbb{C}^{c_{\lambda,\mu}^{\gamma}}$. For the case $\mu = \ydiagram{1}$, we make the same choice as in \cite{Grinko_2023,Nguyen_2023,Vilenkin_1995}. We can then extend this choice in a consistent way to general $\mu$ (see Proposition \ref{prop:multiplicity_irreps}), but the actual choice is irrelevant for the results of this work. Sometimes it is useful to have variable input irreps for a given Clebsch--Gordan transform. For the special cases of fixing $\mu = \ydiagram{1}$ (or $\mu = \overline{\ydiagram{1}}$), we define the simple (dual) Clebsch--Gordan transforms as
\begin{align}
    U_{\rm CG}^{m,n,d} := \bigoplus_{\nu \vdash_{d} (m,n)} U_{\rm CG}^{\nu, \ydiagram{1}, d} \quad , \quad U_{\rm dCG}^{m,n,d} := \bigoplus_{\nu \vdash_{d} (m,n)} U_{\rm CG}^{\nu, \overline{\ydiagram{1}}, d} \, .
\end{align}
We also remark here that for $\gamma \in \nu +_{d} \ydiagram{1}$, there exists $1\leq i \leq d$ so that $\gamma_{i} = \nu_{i} + 1$. This corresponds to appending a box to $\gamma$ in such a way that the result is again a staircase. Similarly, for $\gamma \in \nu +_{d} \overline{\ydiagram{1}}$, there exists $1\leq i \leq d$ so that $\gamma_{i} = \nu_{i} - 1$. This corresponds to removing a box. Furthermore, we always have $c_{\nu, \ydiagram{1}}^{\gamma}, c_{\nu, \overline{\ydiagram{1}}}^{\gamma} \in \{0,1\}$, so there is no multiplicity in these cases.
\newline

\begin{example}
    Let $\gamma \vdash_{4} (4,2)$ be given by $\gamma = (2,1,0,-1)$. Then the set $\gamma+_{d}\ydiagram{1}$ is given by
    \begin{align}
        \gamma+_{d}\ydiagram{1} = \left\{
        \ytableausetup{boxsize=1em}
        \ydiagram{1+2,1+1,1+0,1+0} \, , \, \ydiagram{1+2,1+1,1+1,1} \, , \, \ydiagram{1+2,1+2,1+0,1} \, , \, \ydiagram{1+3,1+1,1+0,1}
        \ytableausetup{boxsize=0.5em}
        \right\} \, .
    \end{align}
\end{example}

\noindent \emph{Equivariant embeddings. } We study the inverse of a Clebsch--Gordan transform, restricted to a specific irrep, also accounting for multiplicity. This yields an isometric embedding that commutes with the action of $SU(d)$. Let $\lambda, \mu, \gamma$ label irreps of $SU(d)$, and let $\ket{\psi}\in \mathbb{C}^{c_{\lambda,\mu}^{\gamma}}$. Then we define
\begin{align}
    \iota_{\lambda,\mu}^{\gamma,\psi} \in \mathcal{B}\left( \mathcal{Q}_{\gamma}^{d} , \mathcal{Q}_{\lambda}^{d} \otimes \mathcal{Q}_{\mu}^{d} \right) \, , \\
    \iota_{\lambda,\mu}^{\gamma,\psi} := (U_{\rm CG}^{\lambda,\mu,d})^{*}|_{\mathcal{Q}_{\gamma}^{d} \otimes \ket{\psi}}
\end{align}
These embeddings will be helpful to understand the action of unitary-equivariant quantum channels from irrep to irrep in Section \ref{subs:operational_interpretation}.

\subsection{Representation theory of $\mathcal{A}_{m,n}^{d}$}

The other main symmetry in this work is given by the group $S_{m}$ of permutations on $m$ distinct elements. This symmetry can be extended to the algebra $\mathcal{A}_{m,n}^{d}$ of partially transposed permutations on $(\mathbb{C}^{d})^{\otimes (m+n)}$, which is important for mixed Schur-Weyl duality. We give a quick overview of the irreps together with a canonical basis. For more information on the representation theory of $\mathcal{A}_{m,n}^{d}$ and the Gel'fand-Tsetlin basis, see \cite{Grinko_2024}.
\newline

\noindent \emph{Permutations and partially transposed permutations. } \noindent For $\sigma \in S_{m}$ we denote by $\sigma \in \mathcal{B}((\mathbb{C}^{d})^{\otimes m})$ the operator permuting the $m$ systems as
\begin{align}
    \sigma = \sum_{i_{1},...,i_{m}} \ketbra{i_{\sigma^{-1}(1)},...,i_{\sigma^{-1}(m)}}{i_{1},...,i_{m}} \, .
\end{align}
Whether $\sigma$ is a group element or an operator will be clear from context. For $\sigma \in S_{m+n}$, the \emph{partially transposed permutation} $\sigma^{\Gamma} \in \mathcal{B}((\mathbb{C}^{d})^{\otimes (m+n)})$ is given by taking the partial transpose on the last $n$ systems, that is
\begin{align}
    \label{equ:partially_transposed_permutations_action}
    \sigma^{\Gamma} = \sum_{i_{1},...,i_{m+n}} \ketbra{i_{\sigma^{-1}(1)},...,i_{\sigma^{-1}(m)},i_{m+1},...,i_{m+n}}{i_{1},...,i_{m},i_{\sigma^{-1}(m+1)},...,i_{\sigma^{-1}(m+n)}} \, .
\end{align}
We define the algebra generated by these partially transposed permutations as $\mathcal{A}_{m,n}^{d}$. The permutations in $S_{m+n}$ are generated by the adjacent swaps $\pi_{i,i+1}$ for $1\leq i \leq m+n-1$. Translating this result to $\mathcal{A}_{m,n}^{d}$, we see that the algebra is generated by the swaps $\pi_{i,i+1}$ for $1\leq i < m$ and $m+1 < i \leq m+n-1$, together with $\pi_{m,m+1}^{\Gamma}$, which is the projection onto the maximally entangled state between sites $m$ and $m+1$. From this intuition we immediately see that the algebra $\mathcal{A}_{m,n}^{d}$ contains the algebras generated by $S_{m}$ acting on $(\mathbb{C}^{d})^{\otimes m}$ and $S_{n}$ acting on $(\mathbb{C}^{d})^{\otimes n}$ respectively as subalgebras. In particular, for $n=0$ we find that $\mathcal{A}_{m,0}^{d}$ is just the algebra generated by $S_{m}$.
\newline

\noindent \emph{Irreps of $\mathcal{A}_{m,n}^{d}$. } The irreps of $S_{m}$ are labeled by staircases $\gamma \vdash_{d} (m,n)$. We denote the corresponding vector space by $\mathcal{P}_{\gamma}^{m,n,d}$, and we fix the basis and action of $\mathcal{A}_{m,n}^{d}$ to be the same as in \cite{Grinko_2023}. Since the matrices involved are all real-valued, the dual action of $\mathcal{A}_{m,n}^{d}$ is the same as the normal action, and we find $\mathcal{P}_{\gamma}^{m,n,d} = \overline{\mathcal{P}_{\gamma}^{m,n,d}}$. For $n=0$, we interpret $\mathcal{P}_{\gamma}^{m,0,d}$ as an irrep of $S_{m}$, which doesn't depend on the dimension $d$, and we simply write $\mathcal{P}_{\gamma}^{m}$.
\newline

\noindent \emph{Gel'fand-Tsetlin basis for $\mathcal{P}_{\gamma}^{m,n,d}$. }
The basis described in \cite{Grinko_2023} has the advantage, that it is adapted to the inclusions  $\mathcal{A}_{m,n-1}^{d} \subseteq \mathcal{A}_{m,n}^{d}$. If we restrict to such a subalgebra, a given irrep $\mathcal{P}_{\gamma}^{m,n,d}$ becomes reducible again and we find
\begin{align}
    \mathcal{P}_{\gamma}^{m,n,d} \stackrel{\mathcal{A}_{m,n-1}^{d}}{\cong} \bigoplus_{\substack{\nu \vdash_{d} (m,n-1): \\ \gamma \in \nu +_{d} \overline{\ydiagram{1}}}} \mathcal{P}_{\nu}^{m,n-1,d} \, .
\end{align}
Similarly, if $n=0$, we have
\begin{align}
    \mathcal{P}_{\gamma}^{m} \stackrel{\mathcal{A}_{m-1,0}^{d}}{\cong} \bigoplus_{\substack{\nu \vdash_{d} m-1: \\ \gamma \in \nu +_{d} \ydiagram{1}}} \mathcal{P}_{\nu}^{m-1} \, .
\end{align}
This means the respective restrictions are \emph{multiplicity-free}. Looking at the following family of subalgebras
\begin{align}
    \mathcal{A}_{1,0}^{d} \subseteq \mathcal{A}_{2,0}^{d} \subseteq ... \subseteq \mathcal{A}_{m,0}^{d} \subseteq \mathcal{A}_{m,1}^{d} \subseteq ... \subseteq \mathcal{A}_{m,n}^{d} \, ,
\end{align}
we get
\begin{align}
    \mathcal{P}_{\gamma}^{m,n,d} \stackrel{\mathcal{A}_{m,n-1}^{d}}{\cong} ... \stackrel{\mathcal{A}_{1,0}^{d}}{\cong} \bigoplus_{\substack{\gamma^{m+n-1} \vdash_{d} (m,n-1): \\ \gamma \in \gamma^{m+n-1} +_{d} \overline{\ydiagram{1}}}} ... \bigoplus_{\substack{\gamma^{1} \vdash_{d} 1: \\ \gamma^{1} \in \gamma^{2} +_{d} \ydiagram{1}}} \mathbb{C} \, .
\end{align}
Here we used the fact that $\mathcal{A}_{1,0}^{d}$ is spanned by the identity and therefore only has trivial, one-dimensional irreps. This is a decomposition of $\mathcal{P}_{\gamma}^{m,n,d}$ into orthogonal, one-dimensional vector space. Each of these subspaces is spanned by a base vector of the \emph{Gel'fand-Tsetlin basis} for $\mathcal{P}_{\gamma}^{m,n,d}$. Each such vector is labeled by a sequence (also called a path) of staircases $(\gamma^{0},...,\gamma^{m+n})$ with the properties
\begin{align}
    &\gamma^{0} = \emptyset \, , \\
    &\gamma^{i} \in \gamma^{i-1} +_{d} \ydiagram{1} \quad \text{for} \quad 1 \leq i \leq m \, , \\
    &\gamma^{i} \in \gamma^{i-1} +_{d} \overline{\ydiagram{1}} \quad \text{for} \quad m+1 \leq i \leq m + n \, , \\
    &\gamma^{m+n} = \gamma \, .
\end{align}
In the case $n=0$, all $\gamma^{i} \vdash_{d} i$ are partitions, and going to $\gamma^{i+1}$ corresponds to adding a box to $\gamma^{i}$.

\begin{example}
    Let $\gamma \vdash_{4} (4,2)$ be given by $\gamma = (2,1,0,-1)$. Then an example for a Gel'fand-Tsetlin base vector of $\mathcal{P}_{\gamma}^{4,2,4}$ is given by $\ket{p_{\gamma}^{4,2,4}}$, where the label $p_{\gamma}^{4,2,4}$ is given by the sequence
    \begin{align}
        p_{\gamma}^{4,2,4} \sim \left(
        \ytableausetup{boxsize=1em}
        \emptyset \, , \, 
        \ydiagram{1,0,0,0} \, , \, \ydiagram{1,1,0,0} \, , \, \ydiagram{2,1,0,0} \, , \, \ydiagram{2,1,1,0} \, , \, \ydiagram{1+2,1+1,1+1,1} \, , \, \ydiagram{1+2,1+1,1+0,1}
        \ytableausetup{boxsize=0.5em}
        \right) \, .
    \end{align}
\end{example}

\subsection{Mixed Schur-Weyl duality}

Mixed Schur–Weyl tells us that we can simultaneously block-diagonalize unitary and partially transposed permutation actions. It can be viewed as a generalization of Schur-Weyl duality, and it replaces the permutation group with the algebra of partially transposed permutations. For more information about Schur-Weyl duality in quantum information, see \cite{Harrow_Th2005}. For more information about mixed Schur-Weyl duality, see \cite{Grinko_2023,Grinko_2024,Nguyen_2023}.
\newline

\noindent \emph{Mixed Schur-Weyl duality. } We take the space $(\mathbb{C}^{d})^{\otimes m} \otimes (\overline{\mathbb{C}^{d}})^{\otimes n}$ with the actions of $SU(d)$ and $\mathcal{A}_{m,n}^{d}$ given by
\begin{align}
    U \mapsto U^{\otimes m} \otimes \overline{U}^{\otimes n} \, \quad \, \sigma^{\Gamma} \mapsto \sigma^{\Gamma} \, .
\end{align}
It is easy to see that $U^{\otimes m} \otimes \overline{U}^{\otimes n}$ commutes with any permutation acting only on $(\mathbb{C}^{d})^{\otimes m}$ or $(\overline{\mathbb{C}^{d}})^{\otimes n}$. In addition, it is easy to verify
\begin{align}
    U \otimes \overline{U} \ketbra{\Psi}{ \Psi} = \ketbra{\Psi}{ \Psi} = \ketbra{\Psi}{ \Psi} U \otimes \overline{U} \, ,
\end{align}
where $\ket{\Psi}$ is the maximally entangled state. As discussed above, the algebra $\mathcal{A}_{m,n}^{d}$ is generated precisely by these permutations and the projection onto the maximally entangled state between sites $m$ and $m+1$. Therefore the actions of $SU(d)$ and $\mathcal{A}_{m,n}^{d}$ commute. In fact, \emph{mixed Schur-Weyl duality} tells us that these actions are maximally commuting, which is equivalent to the following vector space decomposition
\begin{align}
    (\mathbb{C}^{d})^{\otimes m} \otimes (\overline{\mathbb{C}^{d}})^{\otimes n} \stackrel{SU(d) , \mathcal{A}_{m,n}^{d}}{\cong} \bigoplus_{\gamma \vdash_{d} (m,n)} \mathcal{P}_{\gamma}^{m,n,d} \otimes \mathcal{Q}_{\gamma}^{d} \, .
\end{align}
We call the isometry $U_{\rm Sch}^{m,n,d}$ underlying the above equation the \emph{mixed Schur transform}. In general, the mixed Schur transform is uniquely defined up to a basis change on the irreps. This definition is again unique up to a phase on each subspace, and we follow the convention given in \cite{Grinko_2023,Nguyen_2023}. For the case $n=0$, we obtain the better known \emph{Schur-Weyl duality}, given by
\begin{align}
    (\mathbb{C}^{d})^{\otimes m} \stackrel{SU(d) , S_{m}}{\cong} \bigoplus_{\lambda \vdash_{d} m} \mathcal{P}_{\lambda}^{m} \otimes \mathcal{Q}_{\lambda}^{d} \, .
\end{align}
We call the respective isometry $U_{\rm Sch}^{m,0,d}$ the \emph{Schur transform} and we denote it by $U_{\rm Sch}^{m,d}$.

\section{Classification and operational interpretation}
\label{sec:classification_and_operational_interpretation}
\subsection{Classification of all unitary-equivariant and permutation-invariant quantum channels}

\begin{proposition}
    \label{prop:choi_matrix_equivariant}
    Let $\Phi\in \mathcal{C}^{d}(m,n)$, and let $C_{\Phi}$ bet the corresponding Choi matrix. Then $\Phi$ being unitary-equivariant is equivalent to
    \begin{align}
        \label{equ:Choi_matrix_commutes_unitary}
        (\overline{U}^{\otimes m} \otimes U^{\otimes n}) C_{\Phi} = C_{\Phi} (\overline{U}^{\otimes m} \otimes U^{\otimes n})
    \end{align}
    holding for all $U\in SU(d)$. Similarly, $\Phi$ being permutation-invariant is equivalent to
    \begin{align}
        \label{equ:Choi_matrix_commutes_permutation}
        (\overline{\sigma}\otimes \tau) C_{\Phi} = C_{\Phi} (\overline{\sigma}\otimes \tau) \, .
    \end{align}
    for all $\sigma\in S_{m}$ and $\tau\in S_{n}$.
\end{proposition}

\begin{remark}
    As remarked in Section \ref{sec:mathematical_preliminaries_notation}, we have $\overline{\sigma} = \sigma$. We still keep the dual notation to make clear that the action comes from the space $(\overline{\mathbb{C}^{d}})^{\otimes m}$.
\end{remark}

\begin{proof}
    Checkable via Choi formalism.
\end{proof}

Using Proposition \ref{prop:choi_matrix_equivariant} and Schur's lemma, we derive a description of the Choi matrix $C_{\Phi}$ of a unitary-equivariant and permutation-invariant quantum channel $\Phi \in \mathcal{C}_{up}^{d}(m,n)$. Equations \eqref{equ:Choi_matrix_commutes_unitary} and \eqref{equ:Choi_matrix_commutes_permutation} then tell us that $C_{\Phi}$ commutes with certain group actions of $SU(d)$, $S_{m}$ and $ S_{n}$. By decomposing $(\overline{\mathbb{C}^d})^{\otimes m}\otimes(\mathbb{C}^d)^{\otimes n}$ under the joint $\mathrm{SU}(d)\times S_m\times S_n$ action, Schur’s lemma forces the Choi matrix to become block-diagonal, with freedom only in multiplicity spaces.

We remember that $C_{\Phi} \in \mathcal{B}((\overline{\mathbb{C}^{d}})^{\otimes m} \otimes (\mathbb{C}^{d})^{\otimes n})$. By applying Schur-Weyl duality\footnote{We almost apply mixed Schur-Weyl duality, but the algebra generated by the actions of $S_{m}$ and $ S_{n}$ is not the full walled Brauer algebra, so we expect some multiplicity in the resulting irreps.} separately to $(\overline{\mathbb{C}^{d}})^{\otimes m}$ and $(\mathbb{C}^{d})^{\otimes n}$, we obtain
\begin{align}
    (\overline{\mathbb{C}^{d}})^{\otimes m}\otimes (\mathbb{C}^{d})^{\otimes n} \stackrel{SU(d) , S_{m} , S_{n}}{\cong} \bigoplus_{\lambda\vdash_{d} m}\bigoplus_{\mu\vdash_{d} n}\left(\overline{\mathcal{P}_{\lambda}^{m}}\otimes \overline{\mathcal{Q}_{\lambda}^{d}}\right)\otimes\left(\mathcal{P}_{\mu}^{n}\otimes \mathcal{Q}_{\mu}^{d}\right) \, .
\end{align}
We rearrange the individual terms as $\overline{\mathcal{P}_{\lambda}^{m}}\otimes \mathcal{P}_{\mu}^{n} \otimes  \overline{\mathcal{Q}_{\lambda}^{d}} \otimes \mathcal{Q}_{\mu}^{d}$ and use the Clebsch--Gordan transform $U_{\rm CG}^{\overline{\lambda},\mu,d}$ to get
\begin{align}
    \label{equ:irreps_permutation_invariant_equivariant_quantum channels_Choi_matrix}
    (\overline{\mathbb{C}^{d}})^{\otimes m}\otimes (\mathbb{C}^{d})^{\otimes n} \stackrel{SU(d) , S_{m} , S_{n}}{\cong} \bigoplus_{\lambda\vdash_{d} m}\bigoplus_{\mu\vdash_{d} n}\overline{\mathcal{P}_{\lambda}^{m}}\otimes \mathcal{P}_{\mu}^{n}\otimes \left(\bigoplus_{\gamma\in \mu+_{d}\overline{\lambda}}\mathcal{Q}_{\gamma}^{d}\otimes \mathbb{C}^{c^{\gamma}_{\overline{\lambda},\mu}}\right) \, .
\end{align}
The underlying isometry operation here is
\begin{align}
    \label{equ:isometry_Choi_matrix}
    I=\left(\bigoplus_{\lambda\vdash_{d} m}\bigoplus_{\mu\vdash_{d} n}\id_{\overline{\mathcal{P}_{\lambda}^{m}}}\otimes \id_{\mathcal{P}_{\mu}^{n}}\otimes U_{\rm CG}^{\overline{\lambda},\mu,d} \right)(\overline{U_{\rm Sch}^{m,d}}\otimes U_{\rm Sch}^{n,d}) \, .
\end{align}
We are now ready to give a description of $\mathcal{C}_{up}^{d}(m,n)$.

\begin{theorem}
    \label{thm:classification_Cusd}
    Let $\Phi\in \mathcal{C}^{d}(m,n)$, and let $I$ be the isometry defined in Equation \eqref{equ:isometry_Choi_matrix}. Then $\Phi\in \mathcal{C}_{up}^{d}(m,n)$ iff
    \begin{align}
        \label{equ:us_quantum channel_form}
        I C_{\Phi}I^{*}= \bigoplus_{\lambda\vdash_{d} m}\bigoplus_{\mu\vdash_{d} n}\id_{\overline{\mathcal{P}_{\lambda}^{m}}}\otimes \frac{1}{\dim \mathcal{P}_{\mu}^{n}}\id_{\mathcal{P}_{\mu}^{n}} \otimes \left(\bigoplus_{\gamma\in \mu+_{d}\overline{\lambda}}\frac{\dim \mathcal{Q}_{\lambda}^{d}}{\dim \mathcal{Q}_{\gamma}^{d}}\id_{\mathcal{Q}_{\gamma}^{d}}\otimes M_{\gamma,\overline{\lambda},\mu}\right) \, ,
    \end{align}
    with $M_{\gamma,\overline{\lambda},\mu}\in\mathcal{B}(\mathbb{C}^{c^{\gamma}_{\overline{\lambda},\mu}})$ and
    \begin{enumerate}
        \item $\forall \, \gamma,\lambda,\mu: \quad M_{\gamma,\overline{\lambda},\mu}\geq 0$, and
        \item $\forall \, \lambda: \quad\sum_{\mu\vdash_{d} n}\sum_{\gamma\in \mu+_{d}\overline{\lambda}}\tr[M_{\gamma,\overline{\lambda},\mu}]=1$.
    \end{enumerate}
\end{theorem}

\begin{proof}
    First we take $\Phi\in \mathcal{C}_{up}^{d}(m,n)$. We see that $C_{\Phi}\in\mathcal{B}\left((\overline{\mathbb{C}^{d}})^{\otimes m}\otimes (\mathbb{C}^{d})^{\otimes n}\right)$, and by Proposition \ref{prop:choi_matrix_equivariant} it commutes with the actions of $SU(d)$, $S_{m}$ and $S_{n}$. This means we use Equation~\eqref{equ:irreps_permutation_invariant_equivariant_quantum channels_Choi_matrix} together with Schur's lemma to find that
    \begin{align}
        I C_{\Phi}I^{*}= \bigoplus_{\lambda\vdash_{d} m}\bigoplus_{\mu\vdash_{d} n}\id_{\overline{\mathcal{P}_{\lambda}^{m}}}\otimes \id_{\mathcal{P}_{\mu}^{n}} \otimes \left(\bigoplus_{\gamma\in \mu+_{d}\overline{\lambda}}\id_{\mathcal{Q}_{\gamma}^{d}}\otimes A_{\gamma,\overline{\lambda},\mu}\right) \,
    \end{align}
    for some matrices $A_{\gamma,\overline{\lambda},\mu}$. The fact that $C_{\phi}\geq 0$ then tells us that all $A_{\gamma,\overline{\lambda},\mu}\geq 0$. Finally, we want to use
    \begin{align}
        \label{equ:classification_Cusd_trace_preserving}
        \tr_{(\mathbb{C}^{d})^{\otimes n}}[C_{\Phi}]=\id_{(\overline{\mathbb{C}^{d}})^{\otimes m}} \, .
    \end{align}
    Since both $U_{\rm CG}^{\overline{\lambda},\mu,d}$ and $\id_{\mathcal{Q}_{\gamma}^{d}}\otimes A_{\gamma,\overline{\lambda},\mu}$ commute with the action of $SU(d)$, we find that
    \begin{align}
        \tr_{\mathcal{Q}_{\mu}^{d}}[(U_{\rm CG}^{\overline{\lambda},\mu,d})^{*}(\id_{\mathcal{Q}_{\gamma}^{d}}\otimes A_{\gamma,\overline{\lambda},\mu})U_{\rm CG}^{\overline{\lambda},\mu,d}]
    \end{align}
    also commutes with the action of $SU(d)$. By Schur's lemma, this means that
    \begin{align}
        \tr_{\mathcal{Q}_{\mu}^{d}}[(U_{\rm CG}^{\overline{\lambda},\mu,d})^{*}(\id_{\mathcal{Q}_{\gamma}^{d}}\otimes A_{\gamma,\overline{\lambda},\mu})U_{\rm CG}^{\overline{\lambda},\mu,d}]=\tr[A_{\gamma,\overline{\lambda},\mu}]\frac{\dim \mathcal{Q}_{\gamma}^{d}}{\dim \overline{\mathcal{Q}_{\lambda}^{d}}}\id_{\overline{\mathcal{Q}_{\lambda}^{d}}} \, .
    \end{align}
    If we now take the partial trace over the output space, this corresponds to the partial trace over all $\mathcal{P}_{\mu}^{n}\otimes \mathcal{Q}_{\mu}^{d}$. This gives
    \begin{align}
        \label{equ:classification_Cusd_sum_condition}
        \overline{U_{\rm Sch}^{m,d}} \tr_{(\mathbb{C}^{d})^{\otimes n}}[C_{\Phi}] (\overline{U_{\rm Sch}^{m,d}})^{*} = \bigoplus_{\lambda\vdash_{d}m}\id_{\overline{\mathcal{P}_{\lambda}^{m}}}\otimes \id_{\overline{\mathcal{Q}_{\lambda}^{d}}}\left(\sum_{\mu\vdash_{d}n}\sum_{\gamma\in\mu+_{d}\overline{\lambda}}\tr[A_{\gamma,\overline{\lambda},\mu}]\dim \mathcal{P}_{\mu}^{n}\frac{\dim \mathcal{Q}_{\gamma}^{d}}{\dim \overline{\mathcal{Q}_{\lambda}^{d}}}\right) \, .
    \end{align}
    Equation~\eqref{equ:classification_Cusd_trace_preserving} now requires that the sums on the right side equal $1$ for each $\lambda\vdash_{d}m$. Rescaling $A_{\gamma,\overline{\lambda},\mu}$ to $M_{\gamma,\overline{\lambda},\mu}$ together with $\dim \overline{\mathcal{Q}_{\lambda}^{d}} = \dim \mathcal{Q}_{\lambda}^{d}$ gives the desired results. 

    Let now $\Phi$ be given by a Choi matrix as in Equation~\eqref{equ:us_quantum channel_form}. By reversing each argument from before, we see that $C_{\Phi}$ commutes with the action of $SU(d)$, $S_{m}$ and $ S_{n}$. We further see that $C_{\Phi}\geq 0$ and that Equation~\eqref{equ:classification_Cusd_trace_preserving} holds. Therefore we have $\Phi\in \mathcal{C}_{up}^{d}(m,n)$.
\end{proof}

Since all quantum channels in a given set can be written as convex combinations of extremal quantum channels, it is sufficient to study only these extremal quantum channels. On the implementation side, every quantum channel can be implemented as a probabilistic mixture of extremal quantum channels, so we only need to be able to implement the extremal quantum channels efficiently. This leads us to the following corollary.

\begin{corollary}
    \label{cor:classification_extremal_points_Cusd}
    Let $I$ be the isometry defined in Equation \eqref{equ:isometry_Choi_matrix}. The extremal points of $\mathcal{C}_{up}^{d}(m,n)$ are all quantum channels $\Phi$ with a Choi-matrix $C_{\Phi}$ of the form
    \begin{align}
        I C_{\Phi} I^{*} = \bigoplus_{\lambda\vdash_{d} m}\id_{\overline{\mathcal{P}_{\lambda}^{m}}}\otimes \frac{1}{\dim \mathcal{P}_{\mu_{\lambda}}^{n}}\id_{\mathcal{P}_{\mu_{\lambda}}^{n}} \otimes \frac{\dim \mathcal{Q}_{\lambda}^{d}}{\dim \mathcal{Q}_{\gamma_{\lambda}}^{d}}\id_{\mathcal{Q}_{\gamma_{\lambda}}^{d}}\otimes \ketbra{\psi_{\lambda}}{\psi_{\lambda}} \, .
    \end{align}
    Here, $\mu_{\lambda}\vdash_{d}n$ is arbitrary and $\gamma_{\lambda}\in \mu_{\lambda}+_{d}\overline{\lambda}$ for each $\lambda\vdash_{d}m$. Further, we have $\ket{\psi_{\lambda}} \in \mathbb{C}^{c^{\gamma_{\lambda}}_{\overline{\lambda},\mu_{\lambda}}}$.
\end{corollary}

\begin{proof}
    In the proof of Theorem \ref{thm:classification_Cusd}, we have seen that the sums in Equation~\eqref{equ:classification_Cusd_sum_condition} have to equal $1$ for each $\lambda\vdash_{d}m$. Since $M_{\gamma,\overline{\lambda},\mu}\geq 0$, the extremal points that satisfy this property are those, where $\tr[M_{\gamma,\overline{\lambda},\mu}]$ is $1$ for one specific touple $(\gamma_{\lambda},\mu_{\lambda})$ and $0$ everywhere else. This corresponds exactly to the set of states on $\mathbb{C}^{c^{\gamma_{\lambda}}_{\overline{\lambda},\mu_{\lambda}}}$, and the extreme points are known to be the rank-$1$ projections.
\end{proof}

\subsection{Operational interpretation}
\label{subs:operational_interpretation}

While Theorem \ref{thm:classification_Cusd} and Corollary \ref{cor:classification_extremal_points_Cusd} are nice technical results, they don't provide any insight into what these quantum channels do or how to implement them on a quantum computer. We show that each extremal quantum channel given in Corollary \ref{cor:classification_extremal_points_Cusd} can be interpreted as the subsequent application of a sequence of simple quantum channels with a better operational interpretation. To begin, for $\gamma\vdash_{d} (m,n)$ we take the projections
\begin{align}
    &\Pi_{\lambda}^{m,n}\in \mathcal{B}(\bigoplus_{\nu\vdash_{d}m}\mathcal{P}_{\nu}^{m,n,d}\otimes\mathcal{Q}_{\nu}^{d}) \, ,\\
    &\Pi_{\gamma}^{m,n} :=\id_{\mathcal{P}_{\gamma}^{m,n,d}}\otimes\id_{\mathcal{Q}_{\gamma}^{d}} \, .
\end{align}
We now define these simple quantum channels.

\begin{definition}
    The \textit{unitary mixed Schur sampling quantum channel} $\Phi_{\rm USS}^{m,n}$ is defined as
    \begin{align}
        &\Phi_{\rm USS}^{m,n} \in \mathcal{C}\left((\mathbb{C}^{d})^{\otimes m}\otimes (\overline{\mathbb{C}^{d}})^{\otimes n} , \bigoplus_{\nu\vdash_{d}(m,n)}\mathcal{Q}_{\nu}^{d} \right) \, ,\\
        &\Phi_{\rm USS}^{m,n}(A):=\bigoplus_{\gamma\vdash_{d}m}\tr_{\mathcal{P}_{\gamma}^{m,n,d}}[\Pi_{\gamma}^{m,n}U_{\rm Sch}^{m,n,d}A(U_{\rm Sch}^{m,n,d})^{*}\Pi_{\gamma}^{m,n}] \, .
    \end{align}
    For the case of $n=0$, we simply write $\Phi_{\rm USS}^{m}$, and we call the quantum channel the \textit{unitary Schur sampling quantum channel}.
\end{definition}

\begin{remark}
    This quantum channel corresponds to first performing the (mixed) Schur transform, then measuring the irrep label $\gamma$ and finally tracing out the permutation (walled Brauer) register $\mathcal{P}_{\gamma}^{m,n,d}$. Unitary (mixed) Schur sampling is discussed in detail in \cite{CerveroMartin_2024}.
\end{remark}

\begin{proposition}
    \label{prop:dual_USS_quantum channel}
    The adjoint quantum channel $(\Phi_{\rm USS}^{m,n})^{*}$ is given by
    \begin{align}
        &(\Phi_{\rm USS}^{m,n})^{*} \in \mathcal{C} \left(\bigoplus_{\nu\vdash_{d}(m,n)}\mathcal{Q}_{\nu}^{d} , (\mathbb{C}^{d})^{\otimes m}\otimes (\overline{\mathbb{C}^{d}})^{\otimes n} \right) \, ,\\
        &(\Phi_{\rm USS}^{m,n})^{*}(A) =(U_{\rm Sch}^{m,n,d})^{*}\left(\bigoplus_{\nu\vdash_{d}(m,n)}\frac{1}{\dim \mathcal{P}_{\nu}^{m,n,d}}\id_{\mathcal{P}_{\nu}^{m,n,d}}\otimes A_\nu\right)U_{\rm Sch}^{m,n,d} \, ,
    \end{align}
    where $A_{\nu}$ is given by $A|_{\mathcal{Q}_{\nu}^{d}}$ with the output restricted to $\mathcal{Q}_{\nu}^{d}$ as well.
\end{proposition}

\begin{proof}
    Follows immediately from the definition.
\end{proof}

\begin{remark}
    This quantum channel corresponds to first appending the permutation (walled Brauer) register in the maximally mixed state and then performing the inverse (mixed) Schur transform. In addition, the adjoint quantum channel is also the pseudoinverse of $\Phi_{\rm USS}^{m,n}$, which makes $\Phi_{\rm USS}^{m,n}$ and $(\Phi_{\rm USS}^{m,n})^{*}$ into partial isometries. In particular, $(\Phi_{\rm USS}^{m})^{*}\circ\Phi_{\rm USS}^{m}$ is just the projection onto permutation symmetrized states, see Section \ref{subs:state_symmetrization}.
\end{remark}

\begin{definition}
    \label{def:unitary_equivariant_quantum channel_irreps}
    For $\lambda\vdash_{d}m$, $\mu\vdash_{d}n$, $\gamma\in \mu+_{d}\overline{\lambda}$ and $\ket{\psi} \in \mathbb{C}^{c^{\gamma}_{\overline{\lambda},\mu}}$ we define the quantum channel $\Phi_{\lambda,\mu}^{\gamma,\psi}$ implicitly via its Choi matrix $C_{\Phi_{\lambda,\mu}^{\gamma,\psi}}$ as
    \begin{align}
        &\Phi_{\lambda,\mu}^{\gamma,\psi} \in \mathcal{C}\left(\mathcal{Q}_{\lambda}^{d} , \mathcal{Q}_{\mu}^{d} \right) \, ,\\
        \label{equ:unitary_irrep_quantum channel}
        &C_{\Phi_{\lambda,\mu}^{\gamma,\psi}}:=\frac{\dim \mathcal{Q}_{\lambda}^{d}}{\dim \mathcal{Q}_{\gamma_{\lambda}}^{d}}(U_{\rm CG}^{\overline{\lambda},\mu,d})^{*}(\id_{\mathcal{Q}_{\gamma}^{d}}\otimes \ketbra{\psi}{\psi})U_{\rm CG}^{\overline{\lambda},\mu,d} \, .
    \end{align}
\end{definition}

\begin{remark}
    These quantum channels are precisely the extremal points in the set of unitary-equivariant quantum channels $\mathcal{B}(\mathcal{Q}_{\lambda}^{d})\rightarrow \mathcal{B}(\mathcal{Q}_{\mu}^{d})$. They have been investigated for the case of $SU(2)$ \cite{Nuwairan_2013,Aschieri_2024}, and they have different equivalent definitions (compare e.g. Equations 1.34-1.36 in \cite{Aschieri_2024}), which we have summarized in Theorem \ref{thm:equivalence_of_unitary_equivariant_quantum channels}. There, we find that
    \begin{align}
        \Phi_{\lambda,\mu}^{\gamma,\psi}(\rho) = \tr_{\overline{\mathcal{Q}_{\gamma}^{d}}}\left[\iota_{\mu,\overline{\gamma}}^{\lambda,\psi}\rho(\iota_{\mu,\overline{\gamma}}^{\lambda,\psi})^{*}\right] \, .
    \end{align}
    We operationally interpret this as embedding the state via $\iota_{\mu,\overline{\gamma}}^{\lambda,\psi}$ and then tracing out the register $\overline{\mathcal{Q}_{\gamma}^{d}}$.
\end{remark}

\begin{theorem}
    \label{thm:operational_interpretation_extremal_quantum channels}
    Let $\Phi$ be an extremal point of $\mathcal{C}_{up}^{d}(m,n)$, and let $\mu_{\lambda}\vdash_{d}n$, $\gamma_{\lambda}\in \mu_{\lambda}+_{d}\overline{\lambda}$ and $\ket{\psi_{\lambda}}\in\mathbb{C}^{c^{\gamma_{\lambda}}_{\overline{\lambda},\mu_{\lambda}}}$ be as in Corollary \ref{cor:classification_extremal_points_Cusd}. Then $\Phi$ is given by
    \begin{align}
        \label{equ:thm_operational_interpretation_quantum channel}
        \Phi=(\Phi_{\rm USS}^{n})^{*}\circ \left(\bigoplus_{\lambda\vdash_{d}m}\Phi_{\lambda,\mu_{\lambda}}^{\gamma_{\lambda},\psi_{\lambda}}\right) \circ \Phi_{\rm USS}^{m} \, .
    \end{align}
\end{theorem}

\begin{remark}
    The $\lambda \vdash_{d} m$ is a classical output of $\Phi_{\rm USS}^{m}$. Therefore we interpret the resulting quantum channel as first applying $\Phi_{\rm USS}^{m}$, then a conditional application of an irrep-level unitary-equivariant quantum channel on the output of $\Phi_{\rm USS}^{m}$, and then appending a maximally mixed state on the permutation register via $(\Phi_{\rm USS}^{n})^{*}$.
\end{remark}

\begin{proof}
    Let $\Phi'$ be the quantum channel defined in Equation~\eqref{equ:thm_operational_interpretation_quantum channel}. We take a basis labeled by
    \begin{align}
        \ket{\lambda,p_{\lambda}^{m},q_{\lambda}^{d}}=\ket{p_{\lambda}^{m}}\otimes \ket{q_{\lambda}^{d}} \in \mathcal{P}_{\lambda}^{m}\otimes \mathcal{Q}_{\lambda}^{d} \subseteq \bigoplus_{\mu\vdash_{d}m}\mathcal{P}_{\mu}^{m}\otimes \mathcal{Q}_{\mu}^{d} \, .
    \end{align}
    In this basis, we find that
    \begin{align}
        \Phi_{\rm USS}^{m} \left((U_{\rm Sch}^{m,d})^{*} \ketbra{\lambda,p_{\lambda}^{m},q_{\lambda}^{d}}{\tilde{\lambda},\tilde{p}_{\tilde{\lambda}}^{m},\tilde{q}_{\tilde{\lambda}}^{d}}U_{\rm Sch}^{m,d}\right)=0
    \end{align}
    if $\lambda\neq\tilde{\lambda}$ or $p_{\lambda}^{m}\neq \tilde{p}_{\tilde{\lambda}}^{m}$ and, otherwise we get
    \begin{align}
        \Phi_{\rm USS}^{m} \left((U_{\rm Sch}^{m,d})^{*} \ketbra{\lambda,p_{\lambda}^{m},q_{\lambda}^{d}}{\lambda,p_{\lambda}^{m},\tilde{q}_{\lambda}^{d}}U_{\rm Sch}^{m,d}\right)=\ketbra{q_{\lambda}^{d}}{\tilde{q}_{\lambda}^{d}} \, .
    \end{align}
    Further, we obtain
    \begin{align}
        \left(\bigoplus_{\lambda\vdash_{d}m}\Phi_{\lambda,\mu_{\lambda}}^{\gamma_{\lambda},\psi_{\lambda}}\right)\left(\ketbra{q_{\lambda}^{d}}{\tilde{q}_{\lambda}^{d}}\right)= \frac{\dim \mathcal{Q}_{\lambda}^{d}}{\dim \mathcal{Q}_{\gamma_{\lambda}}^{d}} \tr_{\overline{\mathcal{Q}_{\lambda}^{d}}}\left[(U_{\rm CG}^{\overline{\lambda},\mu,d})^{*}(\id_{\mathcal{Q}_{\gamma}^{d}}\otimes \ketbra{\psi_{\lambda}}{\psi_{\lambda}})U_{\rm CG}^{\overline{\lambda},\mu,d}(\ketbra{\tilde{q}_{\lambda}^{d}}{q_{\lambda}^{d}}\otimes\id_{\mathcal{Q}_{\mu}^{d}})\right] \, .
    \end{align}
    Finally, by Proposition \ref{prop:dual_USS_quantum channel}, applying $(\Phi_{\rm USS}^{n})^{*}$ to the result just adds the maximally mixed state $\frac{1}{\dim\mathcal{P}_{\mu}^{n}}\id_{\mathcal{P}_{\mu}^{n}}$ in the permutation register, before performing the inverse Schur transform. Plugging all of the above together, we get the results for $\Phi' \left((U_{\rm Sch}^{m,d})^{*} \ketbra{\lambda,p_{\lambda}^{m},q_{\lambda}^{d}}{\tilde{\lambda},\tilde{p}_{\tilde{\lambda}}^{m},\tilde{q}_{\tilde{\lambda}}^{d}}U_{\rm Sch}^{m,d}\right)$. If we compute the results for $\Phi$ via the Choi formalism, we see that
    \begin{align}
        \Phi' \left((U_{\rm Sch}^{m,d})^{*} \ketbra{\lambda,p_{\lambda}^{m},q_{\lambda}^{d}}{\tilde{\lambda},\tilde{p}_{\tilde{\lambda}}^{m},\tilde{q}_{\tilde{\lambda}}^{d}}U_{\rm Sch}^{m,d}\right)=\Phi \left((U_{\rm Sch}^{m,d})^{*} \ketbra{\lambda,p_{\lambda}^{m},q_{\lambda}^{d}}{\tilde{\lambda},\tilde{p}_{\tilde{\lambda}}^{m},\tilde{q}_{\tilde{\lambda}}^{d}}U_{\rm Sch}^{m,d}\right) \, .
    \end{align}
    Since the vectors $\ket{\lambda,p_{\lambda}^{m},q_{\lambda}^{d}}$ form a basis, this means that $\Phi'=\Phi$.
\end{proof}

\section{General ansatz}
\label{sec:general_ansatz}

Theorem \ref{thm:operational_interpretation_extremal_quantum channels} tells us that we have to implement the quantum channels $\Phi_{\rm USS}^{m}$, $\Phi_{\lambda,\mu}^{\gamma,\psi}$ and $(\Phi_{\rm USS}^{n})^{*}$ to implement any extremal quantum channel in $\mathcal{C}_{up}^{d}(m,n)$. Our approach is based on algorithms for the Schur transform \cite{Bacon_2005} and the mixed Schur transform \cite{Nguyen_2023,Grinko_2024} via repeated simple Clebsch--Gordan transforms. In addition, we use the recent idea of unitary (mixed) Schur sampling \cite{CerveroMartin_2024} to make the algorithms streaming and memory efficient.

We indicate the $\rho_{A}$ on a register $A$ in the circuit diagrams that follow. This just notation and slightly wrong, as we do not refer to the reduced density matrix on the register $A$, and the different registers will in general be entangled.

A depiction of the original circuit diagram of the mixed Schur transform is given in Figure \ref{fig:Schur_transform_diagram}. The input lives on the space $(\mathbb{C}^{d})^{\otimes m}\otimes(\overline{\mathbb{C}^{d}})^{\otimes n}$, with $\gamma^{1}=\ydiagram{1}$. The $\gamma^{i}$ denote the staircase in the $i$-th step, and the $\rho_{\gamma^{i}}$ denote the (possibly still entangled) states living on $\mathcal{Q}_{\gamma^{i}}^{d}$. The output consists of a state on the unitary register $\rho_{\mathcal{Q}_{\gamma}^{d}}$, the final staircase $\gamma$ to keep track of the isotypic subspace and the collection of intermediate staircases $\gamma^{1},...,\gamma^{m+n-1}$, which make up the state on the permutation register $\rho_{\mathcal{P}_{\gamma}^{m,n,d}}$.

\begin{figure}[h!]
    \centering
    \scalebox{1}{\input{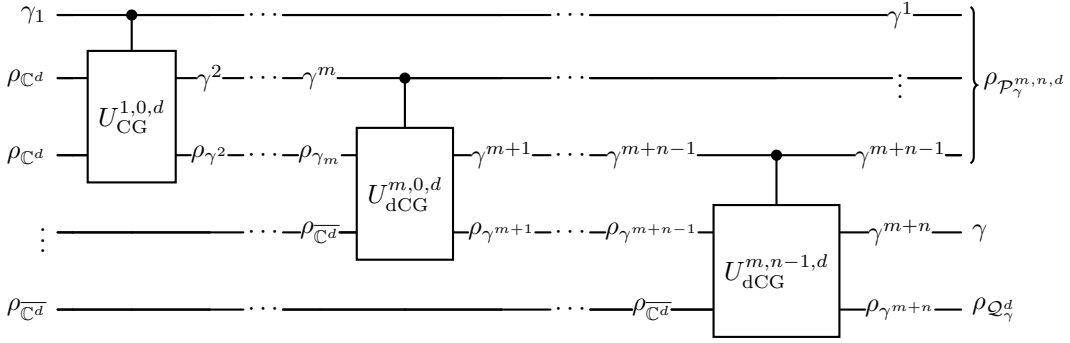}}
    \caption{Implementing the mixed Schur transform \cite{Nguyen_2023,Grinko_2024}.}
    \label{fig:Schur_transform_diagram}
\end{figure}

The idea behind unitary mixed Schur sampling is to supply the states $\rho_{\mathbb{C}^{d}}$ and $\rho_{\overline{\mathbb{C}^{d}}}$ only when they are necessary, and to measure the $\gamma^{i}$ when they're no longer required. The circuit diagram for unitary mixed Schur sampling is given in Figure \ref{fig:USS_diagram}. If we invert all gates in the unitary mixed Schur sampling circuit, we get the process described in Figure \ref{fig:dUSS_diagram}. By picking a uniformly random Gel'fand-Tsetlin base vector $\ket{p_{\gamma}^{m,n,d}}\in \mathcal{P}_{\gamma}^{m,n,d}$ for the state $\rho_{\mathcal{P}_{\gamma}^{m,n,d}}$, we efficiently implement $(\Phi_{\rm USS}^{m,n})^{*}$ in a streaming way.

\begin{figure}[h!]
    \centering
    \scalebox{1}{\input{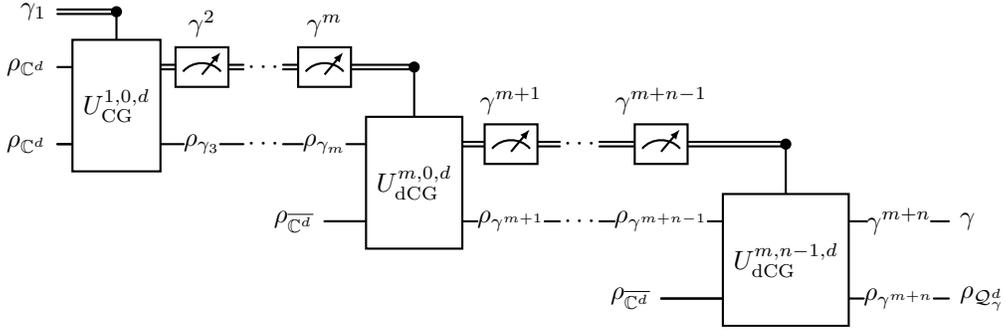}}
    \caption{Implementing unitary mixed Schur sampling \cite{CerveroMartin_2024}.}
    \label{fig:USS_diagram}
\end{figure}

\begin{figure}[h!]
    \centering
    \scalebox{1}{\input{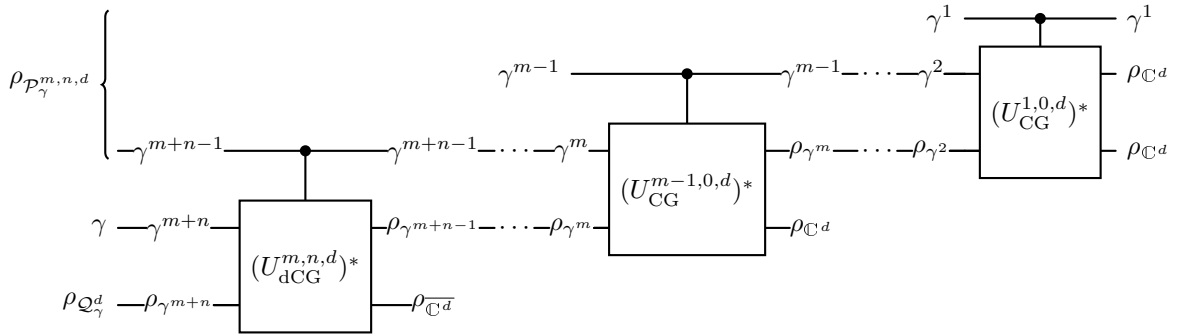}}
    \caption{Implementing the dual quantum channel to unitary mixed Schur sampling.}
    \label{fig:dUSS_diagram}
\end{figure}

Instead of applying the simple dual Clebsch--Gordan transform $n$ times and the simple Clebsch--Gordan transform $m$ times as shown in Figure \ref{fig:dUSS_diagram}, we use an arbitrary number of each, as depicted in Figure \ref{fig:embedding_diagram}. The state $\rho_{\mathcal{P}_{\mu \rightarrow \lambda}^{k,l,d}}$ lives on the vector space of all possible paths of staircases from $\mu$ to $\lambda$ in $(k,l)$ steps. We obtain any isotypic subspace $\mathcal{Q}_{\gamma}^{d}\otimes \mathcal{P}_{\gamma}^{k,l,d}$ on $(\mathbb{C}^{d})^{\otimes k}\otimes(\overline{\mathbb{C}^{d}})^{\otimes l}$, as long as $\lambda \in \mu+_{d}\gamma$, we only need the correct input for $\rho_{\mathcal{P}_{\mu\rightarrow \lambda}^{k,l,d}}$. This lets us embed $\mathcal{Q}_{\lambda}^{d}\hookrightarrow \mathcal{Q}_{\mu}^{d} \otimes \mathcal{Q}_{\gamma}^{d}$ in a unitary-equivariant manner, and the multiplicity of this embedding is reflected in the choice of input $\rho_{\mathcal{P}_{\mu\rightarrow \lambda}^{k,l,d}}$. We make the whole embedding into a streaming algorithm with a similar trick as with unitary Schur sampling (compare Figures \ref{fig:Schur_transform_diagram} and \ref{fig:dUSS_diagram}). According to Theorem \ref{thm:equivalence_of_unitary_equivariant_quantum channels}, we apply the quantum channel $\Phi_{\lambda,\gamma}^{\mu,\psi}$ by implementing this embedding and tracing out $\mathcal{Q}_{\gamma}^{d}$. We also implement the quantum channel $\Phi_{\lambda,\gamma}^{\mu,\psi}$ by tracing out $\mathcal{Q}_{\gamma}^{d}$ instead, but this means wenot perform the algorithm in a streaming manner. 

\begin{figure}[h!]
    \centering
    \scalebox{1}{\input{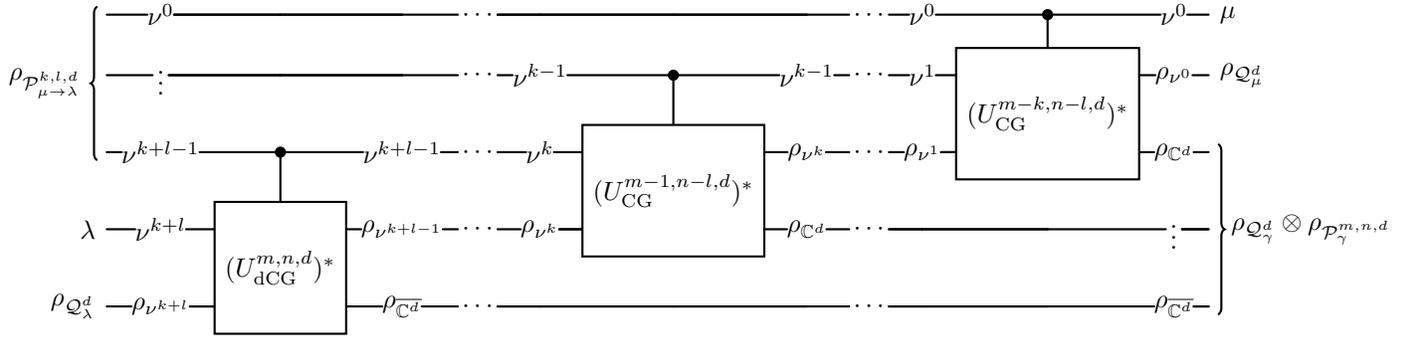}}
    \caption{Implementing the embedding $\iota_{\mu,\gamma}^{\lambda,\psi}$.}
    \label{fig:embedding_diagram}
\end{figure}

\subsection{Iterated simple Clebsch--Gordan transforms over paths}
\label{subs:iterated_simple_CG_transforms}

We realize that a series of simple Clebsch-Gordan transforms corresponds to the direct sum of several general Clebsch-Gordan transforms. We set $\nu^{0}:=\mu$ and get through repeated use of the simple (dual) Clebsch--Gordan transform that
\begin{align}
    \label{equ:iterated_CG_decomposition_simple}
    \mathcal{Q}_{\mu}^{d} \otimes (\mathbb{C}^{d})^{\otimes k} \otimes (\overline{\mathbb{C}^{d}})^{\otimes l} \stackrel{SU(d)}{\cong} \bigoplus_{\nu^{1}\in \nu^{0} +_{d} \ydiagram{1}} ... \bigoplus_{\nu^{k}\in \nu^{k-1} +_{d} \ydiagram{1}} \bigoplus_{\nu^{k+1}\in \nu^{k} +_{d} \overline{\ydiagram{1}}} ... \bigoplus_{\nu^{k+l}\in \nu^{k+l-1} +_{d} \overline{\ydiagram{1}}} \mathcal{Q}_{\nu^{k+l}}^{d} \, .
\end{align}
Each copy of an irrep $\mathcal{Q}_{\lambda}^{d}$ on the right hand side can be associated with a path $p_{\mu \rightarrow \lambda}^{k,l,d}=(\nu^{0},...,\nu^{k+l})$ with the properties
\begin{align}
    &\nu^{0} = \mu \, , \\
    &\nu^{i+1} \in \nu^{i} +_{d} \ydiagram{1} \quad \text{for} \quad i<k \, , \\
    &\nu^{i+1} \in \nu^{i} +_{d} \overline{\ydiagram{1}} \quad \text{for} \quad i \geq k \, , \\
    &\nu^{k+l} = \lambda \, .
\end{align}
Let now $\mu +_{d} (k,l)$ be the set of all labels that appear in Equation~\eqref{equ:iterated_CG_decomposition_simple}, that is the set of all $\lambda$ so that there exists a path $(\nu^{0},...,\nu^{k+l})$ as described above with $\nu^{k+l}=\lambda$. Let further $\mathcal{P}_{\mu \rightarrow \lambda}^{k,l,d}$ be the vector space of formal complex linear combinations of such paths $p_{\mu \rightarrow \lambda}^{k,l,d}$. Using these two definitions, we rewrite Equation~\eqref{equ:iterated_CG_decomposition_simple} to obtain
\begin{align}
    \label{equ:iterated_CG_decomposition_paths}
    \mathcal{Q}_{\mu}^{d} \otimes (\mathbb{C}^{d})^{\otimes k} \otimes (\overline{\mathbb{C}^{d}})^{\otimes l} \stackrel{SU(d)}{\cong} \bigoplus_{\lambda \in \mu +_{d} (k,l)}  \mathcal{P}_{\mu \rightarrow \lambda}^{k,l,d} \otimes \mathcal{Q}_{\lambda}^{d} \, .
\end{align}
We take the isometry behind the equation above as
\begin{align}
    \label{equ:isometry_to_path_space}
    I_{\mu}^{k,l,d} \in \mathcal{B} \left( \mathcal{Q}_{\mu}^{d} \otimes (\mathbb{C}^{d})^{\otimes k} \otimes (\overline{\mathbb{C}^{d}})^{\otimes l} , \bigoplus_{\lambda \in \mu +_{d} (k,l)}  \mathcal{P}_{\mu \rightarrow \lambda}^{k,l,d} \otimes \mathcal{Q}_{\lambda}^{d} \right) \, .
\end{align}
By our construction via the simple Clebsch--Gordan transforms in Equation~\eqref{equ:iterated_CG_decomposition_simple} we take $I_{\mu}^{k,l,d}$ so that it commutes with the action of $SU(d)$. For $\ket{p_{\mu \rightarrow \lambda}^{k,l,d}} \in \mathcal{P}_{\mu \rightarrow \lambda}^{k,l,d}$ and $\ket{q_{\lambda}^{d}} \in \mathcal{Q}_{\lambda}^{d}$ we define the unitary-equivariant isometric embedding
\begin{align}
    \iota_{p_{\mu \rightarrow \lambda}^{k,l,d}}: \mathcal{Q}_{\lambda}^{d} \rightarrow \mathcal{Q}_{\mu}^{d} \otimes (\mathbb{C}^{d})^{\otimes k} \otimes (\overline{\mathbb{C}^{d}})^{\otimes l} \, , \\
    \iota_{p_{\mu \rightarrow \lambda}^{k,l,d}}(\ket{q_{\lambda}^{d}}):=(I_{\mu}^{k,l,d})^{*}\left(\ket{p_{\mu \rightarrow \lambda}^{k,l,d}} \otimes \ket{q_{\lambda}^{d}}\right) \, .
\end{align}

\begin{remark}
    It is important to note here that $\ket{p_{\mu \rightarrow \lambda}^{k,l,d}}$ can in general be a superposition of paths. Given the state $\ket{p_{\mu \rightarrow \lambda}^{k,l,d}}$ on a quantum register, we see in Proposition~\ref{prop:implement_paths_embedding} that $\iota_{p_{\mu \rightarrow \lambda}^{k,l,d}}$ is efficiently implementable, as it just corresponds to a partial inverse Schur transform as described in \cite{Bacon_2005}.
\end{remark}

Another approach to decompose Equation~\eqref{equ:iterated_CG_decomposition_simple} into irreps is to first apply mixed Schur-Weyl duality to $(\mathbb{C}^{d})^{\otimes k} \otimes (\overline{\mathbb{C}^{d}})^{\otimes l}$ and afterwards use the general Clebsch--Gordan transform. This gives us
\begin{align}
    \label{equ:iterated_CG_decomposition_SW_duality}
    \mathcal{Q}_{\mu}^{d} \otimes (\mathbb{C}^{d})^{\otimes k} \otimes (\overline{\mathbb{C}^{d}})^{\otimes l} \stackrel{ SU(d) , \mathcal{A}_{k,l}^{d}}{\cong} \bigoplus_{\gamma \vdash_{d} (k,l)} \mathcal{P}_{\gamma}^{k,l,d} \otimes \mathcal{Q}_{\gamma}^{d} \otimes \mathcal{Q}_{\mu}^{d} \stackrel{SU(d)}{\cong} \bigoplus_{\gamma \vdash_{d} (k,l)} \bigoplus_{\lambda \in \mu +_{d} \gamma} \mathcal{P}_{\gamma}^{k,l,d} \otimes \mathcal{Q}_{\lambda}^{d} \otimes \mathbb{C}^{c_{\mu,\gamma}^{\lambda}} \, .
\end{align}
Comparing Equations~\eqref{equ:iterated_CG_decomposition_paths} and~\eqref{equ:iterated_CG_decomposition_SW_duality}, we see that
\begin{align}
    \label{equ:iterated_CG_decomposition_paths_isometry}
    \mathcal{P}_{\mu \rightarrow \lambda}^{k,l,d} \stackrel{\mathcal{A}_{k,l}^{d}}{\cong} \bigoplus_{\substack{\gamma \vdash_{d} (k,l): \\ \lambda \in \mu +_{d} \gamma}} \mathcal{P}_{\gamma}^{k,l,d} \otimes \mathbb{C}^{c_{\mu,\gamma}^{\lambda}} \, .
\end{align}
We take the isometry behind the equation above as
\begin{align}
    J_{\mu \rightarrow \lambda}^{k,l,d}: \mathcal{P}_{\mu \rightarrow \lambda}^{k,l,d} \rightarrow \bigoplus_{\substack{\gamma \vdash_{d} (k,l): \\ \lambda \in \mu +_{d} \gamma}} \mathcal{P}_{\gamma}^{k,l,d} \otimes \mathbb{C}^{c_{\mu,\gamma}^{\lambda}} \, ,
\end{align}
in a way so that we have the following relation
\begin{align}
    \label{equ:commuting_isometries_relations}
    \left(\bigoplus_{\gamma\vdash_{d}(k,l)}U_{\rm CG}^{\mu,\gamma} \otimes \id_{\mathcal{P}_{\gamma}^{m,n,d}}\right)\left(\id_{\mathcal{Q}_{\mu}^{d}} \otimes U_{\rm Sch}^{k,l,d}\right) = \left(\bigoplus_{\lambda \in \mu +_{d} (k,l)} J_{\mu \rightarrow \lambda}^{k,l,d} \otimes \id_{\mathcal{Q}_{\lambda}^{d}} \right) I_{\mu}^{k,l,d} \, .
\end{align}
This is equivalent to the statement, that the diagram given in Figure~\ref{fig:commuting_isometries_relations} commutes. We now state the following result.

\begin{figure}[h!]
    \centering
    \scalebox{1.1}{\begin{tikzcd}[row sep=2cm, column sep=5cm]
\mathcal{Q}_{\mu}^{d} \otimes (\mathbb{C}^{d})^{\otimes k} \otimes (\overline{\mathbb{C}^{d}})^{\otimes l} \arrow[r, "I_{\mu}^{k,l,d}"] \arrow[d, "\left(\id_{\mathcal{Q}_{\mu}^{d}} \otimes U_{\rm Sch}^{k,l,d}\right)"'] & \bigoplus\limits_{\lambda \in \mu +_{d}(k,l)} \mathcal{P}_{\mu\rightarrow\lambda}^{k,l,d} \otimes \mathcal{Q}_{\lambda}^{d} \arrow[d, "\left(\bigoplus\limits_{\lambda \in \mu +_{d}(k,l)} J_{\mu\rightarrow\lambda}^{k,l,d} \otimes \id_{\mathcal{Q}_{\lambda}^{d}}\right)"] \\
\mathcal{Q}_{\mu}^{d} \otimes \left(\bigoplus\limits_{\gamma\vdash_{d}(k,l)}\mathcal{P}_{\gamma}^{k,l,d} \otimes \mathcal{Q}_{\gamma}^{d}\right) \arrow[r, "\left(\bigoplus\limits_{\gamma\vdash_{d}(k,l,d)}U_{\rm CG}^{\mu,\gamma,d} \otimes \id_{\mathcal{P}_{\gamma}^{k,l,d}}\right)"'] & \bigoplus\limits_{\gamma\vdash_{d}(k,l)}\bigoplus\limits_{\lambda\in \mu +_{d} \gamma}\mathcal{P}_{\gamma}^{k,l,d} \otimes \mathcal{Q}_{\lambda}^{d} \otimes \mathbb{C}^{c_{\mu,\gamma}^{\lambda}}
\end{tikzcd}}
    \caption{Commutative diagram for Equation~\eqref{equ:commuting_isometries_relations}.}
    \label{fig:commuting_isometries_relations}
\end{figure}

\begin{proposition}
    \label{prop:paths_embedding}
    Let $\lambda\vdash_{d} m$ and $\mu\vdash_{d} n$, and let $\gamma \vdash_{d} (k,l)$ so that $\lambda \in \mu +_{d} \gamma$. Let further $\ket{p_{\gamma}^{m,n,d}} \in \mathcal{P}_{\gamma}^{m,n,d}$ and $\ket{\psi} \in \mathbb{C}^{c_{\mu,\gamma}^{\lambda}}$. Finally, we take
    \begin{align}
        \ket{p_{\mu\rightarrow\lambda}^{k,l,d}} := (J_{\mu \rightarrow \lambda}^{k,l,d})^{*}(\ket{p_{\gamma}^{m,n,d}} \otimes \ket{\psi}) \, .
    \end{align}
    Then we have
    \begin{align}
        \ket{p_{\gamma}^{m,n,d}} \otimes \iota_{\mu,\gamma}^{\lambda,\psi} (\ket{q_{\lambda}^{d}}) = (\id_{\mathcal{Q}_{\lambda}^{d}} \otimes U_{\rm Sch}^{k,l,d}) \left( \iota_{p_{\mu\rightarrow\lambda}^{k,l,d}} (\ket{q_{\lambda}^{d}}) \right)
    \end{align}
    for all $\ket{q_{\lambda}^{d}} \in \mathcal{Q}_{\lambda}^{d}$.
\end{proposition}

\begin{remark}
    Proposition~\ref{prop:paths_embedding} tells us that, by choosing the right superposition of paths $\ket{p_{\mu\rightarrow\lambda}^{k,l,d}} \in \mathcal{P}_{\mu\rightarrow\lambda}^{k,l,d}$ and then applying the transformation from right hand side to left hand side in Equation~\eqref{equ:iterated_CG_decomposition_paths}, we get the embedding $\iota_{\mu,\gamma}^{\lambda,\psi}$ (modulo a mixed Schur transform on the $k,l$ copies of $\mathbb{C}^{d}$ and $\overline{\mathbb{C}^{d}}$). However, since Equation~\eqref{equ:iterated_CG_decomposition_paths} came from Equation~\eqref{equ:iterated_CG_decomposition_simple}, we have reduced the problem to a repeated application of inverse simple Clebsch--Gordan transforms, in a superposition of paths given by $\ket{p_{\mu\rightarrow\lambda}^{k,l,d}} = (J_{\mu \rightarrow \lambda}^{k,l,d})^{*}(\ket{p_{\gamma}^{m,n,d}} \otimes \ket{\psi})$.
\end{remark}

\begin{proof}
    We interpret the embedding $\iota_{\mu,\gamma}^{\lambda,\psi}$ as the restriction of $(U_{\rm CG}^{\mu,\gamma})^{*}$ to the subspace $\mathcal{Q}_{\lambda}^{d} \otimes \text{span}(\ket{\psi})$, that is
    \begin{align}
        \label{equ:proof_paths_embedding_1}
        \ket{p_{\gamma}^{m,n,d}} \otimes \iota_{\mu,\gamma}^{\lambda,\psi} (\ket{q_{\lambda}^{d}}) = \ket{p_{\gamma}^{m,n,d}} \otimes(U_{\rm CG}^{\mu,\gamma})^{*}(\ket{q_{\lambda}^{d}}\otimes\ket{\psi}) \, .
    \end{align}
    On the other hand, we have
    \begin{align}
        \iota_{p_{\mu\rightarrow\lambda}^{k,l,d}} (\ket{q_{\lambda}^{d}}) = (I_{\mu}^{k,l,d})^{*}\left(\ket{p_{\mu\rightarrow\lambda}^{k,l,d}} \otimes \ket{q_{\lambda}^{d}}\right) = (I_{\mu}^{k,l,d})^{*}\left((J_{\mu \rightarrow \lambda}^{k,l,d})^{*}(\ket{p_{\gamma}^{m,n,d}} \otimes \ket{\psi}) \otimes \ket{q_{\lambda}^{d}}\right) \, .
    \end{align}
    By the commutativity of the diagram given in Figure~\ref{fig:commuting_isometries_relations}, this means that
    \begin{align}
        \label{equ:proof_paths_embedding_2}
        \iota_{p_{\mu\rightarrow\lambda}^{k,l,d}} (\ket{q_{\lambda}^{d}}) = (\id_{\mathcal{Q}_{\mu}^{d}}\otimes (U_{\rm Sch}^{k,l,d})^{*})\left(\ket{p_{\gamma}^{m,n,d}} \otimes(U_{\rm CG}^{\mu,\gamma})^{*}(\ket{q_{\lambda}^{d}}\otimes\ket{\psi})\right) \, .
    \end{align}
    Comparing Equations~\eqref{equ:proof_paths_embedding_1} and \eqref{equ:proof_paths_embedding_2} gives the proposition.
\end{proof}

To show that $\iota_{p_{\mu\rightarrow\lambda}^{k,l,d}}$ is efficiently implementable, we need an efficient implementation of the simple (dual) Clebsch--Gordan transforms, which was first given in \cite{Bacon_2005}, with an improvement in gate complexity given in \cite{CerveroMartin_2024}.

\begin{proposition}
    \label{prop:superposition_CG_transforms}
    Let $1\leq r \leq m+n$. Then the unitaries given by
    \begin{align}
        U_{\rm CG}^{m,n,d,r} := \bigoplus_{\substack{\gamma\vdash_{d} (m,n) \\ l(\gamma) \leq r}} U_{\rm CG}^{\gamma, \, \ydiagram{1}} \, , \\
        U_{\rm dCG}^{m,n,d,r} := \bigoplus_{\substack{\gamma\vdash_{d} (m,n) \\ l(\gamma) \leq r}} U_{\rm CG}^{\gamma, \, \overline{\ydiagram{1}}} \, ,
    \end{align}
    can be implemented up to error $\epsilon$ in operator norm with memory complexity $M$ and gate complexity $T$, where
    \begin{align}
        &M=O\big(rd\log_{2}^{p}(d,m,n,1/\epsilon)\big) \, ,\\
        \label{equ:gate_complexity_reduced}
        &T=O\big(r^3d\log_{2}^{p}(d,m,n,1/\epsilon)\big) \, ,
    \end{align}
    with $p\approx 1.44$.
\end{proposition}

\begin{remark}
    In essence, Proposition~\ref{prop:superposition_CG_transforms} tells us that we efficiently perform the simple (dual) Clebsch--Gordan transforms in superposition, conditioned on a quantum register labeling the staircase $\gamma$, and with an additional improvement in memory and y when we have the promise $l(\gamma)\leq r$.
\end{remark}

\begin{proof}
    See the analysis in Section IV.C of \cite{CerveroMartin_2024}.
\end{proof}

Proposition~\ref{prop:superposition_CG_transforms} now gives an easy algorithm for efficiently implementing the embeddings $\iota_{p_{\mu\rightarrow\lambda}^{k,l,d}}$.

\begin{proposition}
    \label{prop:implement_paths_embedding}
    Let $\lambda\vdash_{d} m$ and $\mu\vdash_{d} n$. Given the resource state $\ket{p_{\mu\rightarrow\lambda}^{k,l,d}}$, the embedding $\iota_{p_{\mu\rightarrow\lambda}^{k,l,d}}$ can be implemented up to error $\epsilon$ in operator norm with memory complexity $M$ and gate complexity $T$, where
    \begin{align}
        &M=O\big(dr log_{2}^{p}(d,m,n,k,l,1/\epsilon)\big) + M_{p_{\mu\rightarrow\lambda}^{k,l,d}} \, ,\\
        \label{equ:gate_complexity_reduced}
        &T=O\big((k+l)dr^{3} \log_{2}^{p}(d,m,n,k,l,1/\epsilon)\big) + T_{p_{\mu\rightarrow\lambda}^{k,l,d}} \, ,
    \end{align}
    with streaming output on the register $(\mathbb{C}^{d})^{\otimes k} \otimes (\overline{\mathbb{C}^{d}})^{\otimes l}$. In the above equation, we have $p\approx 1.44$, $r=\max\limits_{0\leq i \leq k+l}l(\nu^{i})$, $M_{p_{\mu\rightarrow\lambda}^{k,l,d}}$ is the memory required to store $\ket{p_{\mu\rightarrow\lambda}^{k,l,d}}$ and $T_{p_{\mu\rightarrow\lambda}^{k,l,d}}$ are the gates required to retrieve the entries of $\nu^i$ from the encoding of $p_{\mu\rightarrow\lambda}^{k,l,d} = (\nu^0,...,\nu^{k+l})$.
\end{proposition}

\begin{remark}
    If we store the resource state $\ket{p_{\mu\rightarrow\lambda}^{k,l,d}}$ with the same encoding we use for all other algorithms, that is we store each individual staircase $\nu$ as a sequence of $d$ integers with $r$ nonzero entries (see \cite{Harrow_Th2005,Grinko_2023,Nguyen_2023}), then we get
    \begin{align}
        M_{p_{\mu\rightarrow\lambda}^{k,l,d}} = O((k+l)r\log_{2}(m+k,n+l)) \quad , \quad T_{p_{\mu\rightarrow\lambda}^{k,l,d}}=0 \, .
    \end{align}
\end{remark}

\begin{remark}
    Together, Propositions~\ref{prop:paths_embedding} and~\ref{prop:implement_paths_embedding} give an algorithm for implementing the embedding $\iota_{\mu,\overline{\gamma}}^{\lambda,\psi}$, which in turn can be used to implement the unitary-equivariant quantum channel $\Phi_{\lambda,\mu}^{\gamma,\psi}$. The only open problem is then finding and initializing the correct vector $\ket{p_{\mu\rightarrow\lambda}^{k,l,d}}$.
\end{remark}

\begin{proof}
    Proposition \ref{prop:superposition_CG_transforms} tells us that $U_{\rm CG}^{m,n,d,r}$ and $U_{\rm dCG}^{m,n,d,r}$ are efficiently implementable with memory and gate complexities $M$ and $T$. Therefore, the isometry $I_{\mu}^{k,l,d}$, given by repeated Clebsch--Gordan transforms as
    \begin{align}
        I_{\mu}^{k,l,d} = U_{\rm dCG}^{m+k,n+l-1,r} ... \left(\id_{\overline{\mathbb{C}^{d}}}^{\otimes l} \otimes \id_{\mathbb{C}^{d}}^{\otimes k-1} \otimes U_{\rm CG}^{m,n,d,r} \right),
    \end{align}
    is efficiently implementable as well, with memory and gate complexities $M$ and $(k+l)T$. The inverse $(I_{\mu}^{k,l,d})^{*}$ has the same memory and gate complexities by simply inverting the elementary gates and applying them in reverse order. By definition, we have
    \begin{align}
        \iota_{p_{\mu \rightarrow \lambda}^{k,l,d}}(\ket{q_{\lambda}^{d}}):=(I_{\mu}^{k,l,d})^{*}\left(\ket{p_{\mu \rightarrow \lambda}^{k,l,d}} \otimes \ket{q_{\lambda}^{d}}\right) \, .
    \end{align}
    We still need to rescale $\epsilon \rightarrow \epsilon / (k+l)$ to keep the total error below $\epsilon$. This introduces terms of the size $\text{poly}(k,l)$ into the logarithm.
\end{proof}

\subsection{General ansatz}

According to Theorem \ref{thm:operational_interpretation_extremal_quantum channels}, we implement an extremal quantum channel $\Phi\in \mathcal{C}_{up}^{d}(m,n)$ via the quantum channels $\Phi_{\rm USS}^{m}$, $(\Phi_{\rm USS}^{n})^{*}$ and $\Phi_{\lambda,\mu}^{\gamma,\rho}$. We first want to apply unitary Schur sampling to obtain an irrep label $\lambda$, together with a state on the unitary irrep $\mathcal{Q}_{\lambda}^{d}$. Conditioned on $\lambda$, we apply the quantum channel $\Phi_{\lambda,\mu_{\lambda}}^{\gamma_{\lambda},\rho_{\lambda}}$ to obtain a state on the unitary irrep $\mathcal{Q}_{\mu_{\lambda}}^{d}$. Finally, we apply the adjoint unitary Schur sampling quantum channel. The relevant circuit diagram is given in Figure \ref{fig:factoring_circuit_diagram}. The remaining question is now how to efficiently implement the quantum channels $\Phi_{\rm USS}^{m}$, $(\Phi_{\rm USS}^{n})^{*}$ and $\Phi_{\lambda,\mu}^{\gamma,\rho}$.

\begin{figure}[h!]
    \centering
    \scalebox{1}{\input{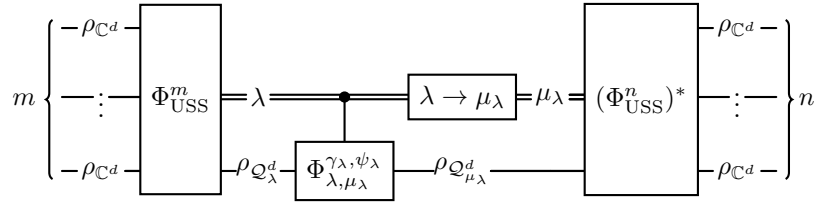}}
    \caption{Circuit diagram showing the factoring of quantum channels in Theorem \ref{thm:operational_interpretation_extremal_quantum channels}.}
    \label{fig:factoring_circuit_diagram}
\end{figure}

\begin{proposition}[Theorem 1 of \cite{CerveroMartin_2024}]
    \label{prop:unitary_Schur_sampling_implementation}
    Let $\sigma$ be a state on $S^{\otimes m}\subseteq (\mathbb{C}^{d})^{\otimes m}$, with $\dim S=r$. Then we implement the operation $\sigma\rightarrow\Phi_{\rm USS}^{m}(\sigma)$ in a streaming manner up to error $\epsilon$ in diamond norm with memory complexity $M$ and gate complexity $T$, where
    \begin{align}
        &M=O\big(rd\log_{2}^{p}(d,m,1/\epsilon)\big) \, ,\\
        \label{equ:gate_complexity_reduced}
        &T=O\big(mr^3d\log_{2}^{p}(d,m,1/\epsilon)\big) \, ,
    \end{align}
    with $p\approx 1.44$.
\end{proposition}

\begin{proposition}
    \label{prop:symmetrization_quantum channel_implementation}
    Let $\sigma \in \bigoplus_{\lambda\vdash_{r'}n}\mathcal{B}(\mathcal{Q}_{\lambda}^{d})$. Then we implement the operation $\sigma\rightarrow(\Phi_{\rm USS}^{n})^{*}(\sigma)$ in a streaming manner up to error $\epsilon$ in diamond norm with memory complexity $M$ and gate complexity $T$, where
    \begin{align}
        &M=O\big(r'd\log_{2}^{p}(d,n,1/\epsilon)\big) \, ,\\
        \label{equ:gate_complexity_reduced}
        &T=O\big(n(r')^3d\log_{2}^{p}(d,n,1/\epsilon)\big) \, ,
    \end{align}
    with $p\approx 1.44$.
\end{proposition}

\begin{proof}
    With the encoding of $\sigma$ given in \cite{Bacon_2005} and used throughout this work, the label $\lambda$ can be measured without extra memory or gate complexity. This reduces the problem to applying $(\Phi_{\rm USS}^{n})^{*}(\rho)$ to a state $\rho\in\mathcal{Q}_{\mu}^{d}$. We see that
    \begin{align}
        \label{equ:proof_symmetrization_quantum channel_implementation_1}
        (\Phi_{\rm USS}^{n})^{*}(\rho) = (U_{\rm Sch}^{n,d})^{*} \left(\frac{1}{\dim \mathcal{P}_{\mu}^{n}}\id_{\mathcal{P}_{\mu}^{n}} \otimes \rho\right) U_{\rm Sch}^{n,d} = \frac{1}{\dim \mathcal{P}_{\mu}^{n}}\sum_{p_{\mu}^{n}\in\mathcal{P}_{\mu}^{n}} (U_{\rm Sch}^{n,d})^{*} \left(\ketbra{p_{\mu}^{n}}{p_{\mu}^{n}} \otimes \rho\right) U_{\rm Sch}^{n,d} \, ,
    \end{align}
    where the sum is over the Gel'fand - Tsetlin basis of $\mathcal{P}_{\mu}^{n}$. By \cite{Bacon_2005}, we know that $U_{\rm Sch}^{n,d}$ is equivalent to $n$ repeated Clebsch--Gordan transforms. In this sense, we interpret the Schur transform as a special case of Equation~\eqref{equ:iterated_CG_decomposition_paths} with $\mu=\emptyset$, $k=n$ and $l=0$, and we see that $U_{\rm Sch}^{n,d}=I_{\emptyset \rightarrow \lambda}^{0,n,d}$. In particular, we see that $\mathcal{P}_{\emptyset \rightarrow \lambda}^{0,n,d} = \mathcal{P}_{\lambda}^{m}$ by simply identifying each path with the basis element of $\mathcal{P}_{\lambda}^{m}$ labeled by that path. By the definition of $\iota_{p_{\mu}^{n}}$ and the previous considerations, we get
    \begin{align}
        (U_{\rm Sch}^{n,d})^{*} \left(\ketbra{p_{\mu}^{n}}{p_{\mu}^{n}} \otimes \rho\right) U_{\rm Sch}^{n,d} = \iota_{p_{\mu}^{n}} \rho (\iota_{p_{\mu}^{n}})^{*} \, .
    \end{align}
    Inserting this into Equation~\eqref{equ:proof_symmetrization_quantum channel_implementation_1} above, we get
    \begin{align}
        (\Phi_{\rm USS}^{n})^{*}(\rho) = \frac{1}{\dim \mathcal{P}_{\mu}^{n}}\sum_{p_{\mu}^{n}\in\mathcal{P}_{\mu}^{n}} \iota_{p_{\mu}^{n}} \rho (\iota_{p_{\mu}^{n}})^{*} \, .
    \end{align}
    This is just a uniform probabilistic mixture of all possible embeddings. Theorem \ref{thm:algorithm_sample_GT_basis} now tells us that we classically sample $p_{\mu}^{n}\in\mathcal{P}_{\mu}^{n}$ in a streaming manner. This means that if $p_{\mu}^{n}=(\nu^{0},...,\nu^{n})$, then we only need to store $\nu^{i+1}$ and $\nu^{i}$ at the same time. In the context of Proposition \ref{prop:implement_paths_embedding} this means that $M_{p_{\mu}^{n}}=O(d\log_{2}(n))$ and $T_{p_{\mu}^{n}}=0$. In addition, Proposition \ref{prop:complexity_sample_one_box_removed_alternative} tells us that the classical complexity of sampling $p_{\mu}^{n}$ is given by $O(nr^{2})$. Applying Proposition \ref{prop:implement_paths_embedding} now gives the result.
\end{proof}

Together, Propositions \ref{prop:unitary_Schur_sampling_implementation} and \ref{prop:symmetrization_quantum channel_implementation} give us the following Theorem.

\begin{theorem}
    \label{thm:complexity_unitary_equivariant_permutation_invariant_quantum channels}
    Let $\Phi$ be an extremal point of $\mathcal{C}_{up}^{d}(m,n)$, and let $\mu_{\lambda}\vdash_{d}n$, $\gamma_{\lambda}\in \mu_{\lambda}+_{d}\overline{\lambda}$ and $\ket{\psi_{\lambda}}\in \mathbb{C}^{c^{\gamma_{\lambda}}_{\overline{\lambda},\mu_{\lambda}}}$ be as in Corollary \ref{cor:classification_extremal_points_Cusd}. Let further $\sigma$ be a state on $S^{\otimes m}\subseteq (\mathbb{C}^{d})^{\otimes m}$, with $\dim S=r$, and let $l(\mu_{\lambda})\leq r'$ for all $\lambda$ with $l(\lambda)\leq r$. Then we implement the operation $\rho\rightarrow\Phi(\rho)$ in a streaming manner up to error $\epsilon$ in diamond norm with memory complexity $M$ and gate complexity $T$, where
    \begin{align}
        &M=O\big((r+r')d\log_{2}^{p}(d,m,n,1/\epsilon)\big) + \max\limits_{\lambda\vdash_{d} m} M_{\mu_{\lambda},\overline{\gamma_{\lambda}}}^{\lambda,\psi_{\lambda}} \, ,\\
        \label{equ:gate_complexity_reduced}
        &T=O\big((mr^{3}+n(r')^{3})d\log_{2}^{p}(d,m,n,1/\epsilon)\big) + \max\limits_{\lambda\vdash_{d} m} T_{\mu_{\lambda},\overline{\gamma_{\lambda}}}^{\lambda,\psi_{\lambda}} \, .
    \end{align}
    Here, $p\approx 1.44$, and $M_{\mu,\overline{\gamma}}^{\lambda,\psi}$ and $T_{\mu,\overline{\gamma}}^{\lambda,\psi}$ denote the memory and gates necessary to implement the embedding $\iota_{\mu,\overline{\gamma}}^{\lambda,\psi}$ and tracing out $\overline{\mathcal{Q}_{\gamma}^{d}}$.
\end{theorem}

\begin{proof}
    Theorem \ref{thm:operational_interpretation_extremal_quantum channels} and Figure \ref{fig:factoring_circuit_diagram} show us that we only need to implement the quantum channels $\Phi_{\rm USS}^{m}$, $(\Phi_{\rm USS}^{n})^{*}$ and $\Phi_{\lambda,\mu_{\lambda}}^{\gamma_{\lambda},\psi_{\lambda}}$ for some $\lambda$. The combined memory and gate complexity is
    \begin{align}
        M = \max(M(\Phi_{\rm USS}^{m}), M(\Phi_{\lambda,\mu_{\lambda}}^{\gamma_{\lambda},\psi_{\lambda}}), M((\Phi_{\rm USS}^{n})^{*})) \, , \\
        T = T(\Phi_{\rm USS}^{m}) + T(\Phi_{\lambda,\mu_{\lambda}}^{\gamma_{\lambda},\psi_{\lambda}}) + T((\Phi_{\rm USS}^{n})^{*}) \, ,
    \end{align}
    where $M(\Psi)$ and $T(\Psi)$ denote the memory and gate complexities of the quantum channel $\Psi$. According to Theorem \ref{thm:equivalence_of_unitary_equivariant_quantum channels}, we implement $\Phi_{\lambda,\mu_{\lambda}}^{\gamma_{\lambda},\psi_{\lambda}}$ by implementing the embedding $\iota_{\mu_{\lambda},\overline{\gamma_{\lambda}}}^{\lambda,\psi_{\lambda}}$ and then discarding the register $\overline{\mathcal{Q}_{\gamma}^{d}}$, which gives us
    \begin{align}
        M(\Phi_{\lambda,\mu_{\lambda}}^{\gamma_{\lambda},\psi_{\lambda}}) = M_{\mu_{\lambda},\overline{\gamma_{\lambda}}}^{\lambda,\psi_{\lambda}} \quad , \quad T(\Phi_{\lambda,\mu_{\lambda}}^{\gamma_{\lambda},\psi_{\lambda}}) = T_{\mu_{\lambda},\overline{\gamma_{\lambda}}}^{\lambda,\psi_{\lambda}} \, .
    \end{align}
    For the memory and ies of $\Phi_{\rm USS}^{m}$ and $(\Phi_{\rm USS}^{n})^{*}$, we apply Propositions \ref{prop:unitary_Schur_sampling_implementation} and \ref{prop:symmetrization_quantum channel_implementation}. For the restriction given by $r'$, we use the fact (\cite[Lemma 14]{CerveroMartin_2024}) that
    \begin{align}
        U_{\rm Sch}^{m,n,d}(S^{\otimes m})\subseteq \bigoplus_{\mu\vdash_{r}m} \mathcal{P}_{\mu}^{n}\otimes\mathcal{Q}_{\mu}^{d} \, .
    \end{align}
    This implies $\Tr[\rho\Pi_{\lambda}^{m}]=0$ for all $\lambda$ with $l(\lambda)>r$. Therefore the input to the quantum channel $(\Phi_{\rm USS}^{n})^{*}$ will have support only on irreps $\mathcal{Q}_{\mu_{\lambda}}^{d}$ for $l(\lambda)\leq r$. However, by assumption this implies in turn that $l(\mu_{\lambda})\leq r'$, which means we use Proposition \ref{prop:symmetrization_quantum channel_implementation} with $r'$.
\end{proof}

\section{Three example applications}
\label{sec:examples}

Theorem \ref{thm:complexity_unitary_equivariant_permutation_invariant_quantum channels} together with Propositions \ref{prop:paths_embedding} and \ref{prop:implement_paths_embedding} now give us a recipe for implementing quantum channels in $\mathcal{C}_{up}^{d}(m,n)$.

\subsection{State symmetrization}
\label{subs:state_symmetrization}

Consider a protocol with $m$ qudits where we want to randomize their ordering. This can be due to implementing an information theoretic protocol, scrambling the information contained in the ordering or averaging errors between the $m$ registers. We describe this task via the following definition.

\begin{definition}[State symmetrization]
    The task of \textit{state symmetrization} is given as the task of implementing the \textit{permutation symmetrization quantum channel} $\Phi_{\rm sym}^{m}$, which is defined as
    \begin{align}
        &\Phi_{\rm sym}^{m}:\mathcal{B}((\mathbb{C}^{d})^{\otimes m}) \rightarrow \mathcal{B}((\mathbb{C}^{d})^{\otimes m}) \, , \\
        &\Phi_{\rm sym}^{m}(\rho) := \frac{1}{m!}\sum_{\sigma \in S_{m}} \sigma \rho \sigma^{*} \, .
    \end{align}
\end{definition}

The permutation symmetrization quantum channel is the simplest example of a unitary-equivariant and permutation-invariant quantum channel. With the framework developed in the previous chapters we obtain the following corollary of Theorem~\ref{thm:complexity_unitary_equivariant_permutation_invariant_quantum channels}.

\begin{corollary}
    \label{cor:result_state_symmetrization}
    Let $\rho$ be a state on $S^{\otimes m} \subseteq (\mathbb{C}^{d})^{\otimes m}$ with $\dim S = r$. Then we implement the operation $\rho\rightarrow\Phi_{\rm sym}^{m}(\rho)$ in a streaming manner up to error $\epsilon$ in diamond norm with memory complexity $M$ and gate complexity $T$, where
    \begin{align}
        &M=O\big(rd\log_{2}^{p}(d,m,1/\epsilon)\big) \, ,\\
        \label{equ:gate_complexity_reduced}
        &T=O\big(mr^{3}d\log_{2}^{p}(d,m,1/\epsilon)\big) \, .
    \end{align}
    Here, $p\approx 1.44$.
\end{corollary}

\begin{proof}
    By using Schur's lemma, we see that
    \begin{align}
        U_{\rm Sch}^{m,d}\Phi_{\rm sym}^{m}(\rho)(U_{\rm Sch}^{m,d})^{*} = \bigoplus_{\lambda\vdash_{d}m} \frac{1}{\dim \mathcal{P}_{\lambda}^{m}} \id_{\mathcal{P}_{\lambda}^{m}} \otimes \tr_{\mathcal{P}_{\lambda}^{m}}[\Pi_{\lambda}^{m,d}U_{\rm Sch}^{m,d}\rho(U_{\rm Sch}^{m,d})^{*}\Pi_{\lambda}^{m,d}] \, .
    \end{align}
    Therefore it's easy to see that $\Phi_{\rm sym}^{m}=\Phi_{\rm USS}^{m} \circ (\Phi_{\rm USS}^{m})^{*}$, and we have $\Phi_{\rm sym}^{m}\in \mathcal{C}_{up}^{d}(m,m)$. Further, we find that $\Phi_{\rm sym}^{m}$ is an extremal point in $\mathcal{C}_{up}^{d}(m,m)$ where all $\Phi_{\lambda,\mu_{\lambda}}^{\gamma_{\lambda},\psi_{\lambda}}$ correspond to the identity quantum channels and are therefore given by $\mu_{\lambda}=\lambda$ and $\gamma_{\lambda}=\emptyset$ (with no multiplicity, so $\psi_{\lambda}=1$). The quantum channels $\Phi_{\lambda,\lambda}^{\emptyset,1}$ have no gate complexity and their memory complexity is just the storage of the register $\mathcal{Q}_{\lambda}^{d}$. The corollary now follows directly from Theorem~\ref{thm:complexity_unitary_equivariant_permutation_invariant_quantum channels}.
\end{proof}

Figure \ref{fig:example_path_symmetrizing} contains an example for the paths appearing in the implementation of $\Phi_{\rm sym}^{m}$.

\begin{figure}[h]
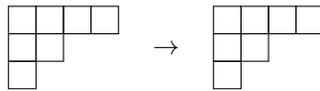

    \begin{align*}
        \ytableausetup{boxsize=1em}
        \ydiagram{4,2,1} \quad \rightarrow \quad \ydiagram{4,2,1}
        \ytableausetup{boxsize=0.5em}
    \end{align*}
    \caption{
    \label{fig:example_path_symmetrizing}
    Path used to implement $\Phi_{\rm sym}^{7}$ for a given $\mu=(4,2,1)$. The path corresponds to the identity quantum channel, so there is no tracing out.}
\end{figure}

\subsection{Symmetric cloning}

In many information theoretic tasks, the ability to copy information is either very desirable or problematic. In the quantum setting, the no-cloning theorem prevents creation of perfect copies of arbitrary quantum states. However, approximate cloning is possible, if we allow for a small error on the clones. A review of results and applications for optimal cloning is given in \cite{Fan_2014}.

\begin{definition}[Symmetric cloning]
    Let $0<m<n$. The task of \textit{symmetric cloning} is given as the task of implementing the \textit{optimal symmetric cloning quantum channel} $\Phi_{\rm cl}^{m,n} \in \mathcal{C}^{d}(m,n)$, which is defined as the quantum channel that maximizes the $m \rightarrow n$ cloning fidelity
    \begin{align}
        \mathcal{F}_{cl}(\Phi) := \min\limits_{\ket{\psi}\in \mathbb{C}^{d}} \tr[\ketbra{\psi}{\psi}^{\otimes n}\Phi(\ketbra{\psi}{\psi}^{\otimes m})] \, .
    \end{align}
\end{definition}

Let $(\mathbb{C}^{d})_{\rm sym}^{\otimes k}$ be the symmetric subspace of $(\mathbb{C}^{d})^{\otimes k}$, and let $\Pi_{\rm sym}^{k,d}$ be the projection onto $(\mathbb{C}^{d})_{\rm sym}^{\otimes k}$. Let further $d[k]$ be the dimension of $(\mathbb{C}^{d})_{\rm sym}^{\otimes k}$. In \cite{Werner_1998}, the quantum channel $\Phi_{\rm cl}^{m,n}$ is given as
\begin{align}
    \Phi_{\rm cl}^{m,n}(\rho) = \frac{d[m]}{d[n]} \Pi_{\rm sym}^{n,d}(\rho \otimes \id_{d}^{\otimes (n-m)}) \Pi_{\rm sym}^{n,d} \, .
\end{align}
In general, this expression is not trace preserving and therefore not a quantum channel. However, the task of symmetric cloning lets us restrict the input to $(\mathbb{C}^{d})_{\rm sym}^{\otimes m}$, where the above expression is actually trace preserving. We now apply the algorithms developed in this work to the above result to get the following corollary.

\begin{corollary}
    \label{cor:result_symmetric_cloning}
    Let $1<m<n$. Given $m$ copies of the state $\ket{\psi}\in \mathbb{C}^{d}$, we implement the operation $\ketbra{\psi}{\psi}^{\otimes m} \rightarrow \Phi_{\rm cl}^{m,n}(\ketbra{\psi}{\psi}^{\otimes m})$ in a streaming manner up to error $\epsilon$ in diamond norm with memory complexity $M$ and gate complexity $T$, where
    \begin{align}
        &M=O\big(d\log_{2}^{p}(d,n,1/\epsilon)\big) \, ,\\
        \label{equ:gate_complexity_reduced}
        &T=O\big(nd\log_{2}^{p}(d,n,1/\epsilon)\big) \, .
    \end{align}
    Here, $p\approx 1.44$.
\end{corollary}

\begin{proof}
    Let $(k,\underline{0}) \vdash_{d} k$ be the partition given by $(k,0,...,0)$. Then we know that $(\mathbb{C}^{d})_{\rm sym}^{\otimes k} \stackrel{SU(d)}{\cong} \mathcal{Q}_{(k,\underline{0})}^{d}$, and we apply Theorem~\ref{thm:equivalence_of_unitary_equivariant_quantum channels} to see that
    \begin{align}
        \Phi_{\rm cl}^{m,n} = (\Phi_{\rm USS}^{n})^{*} \circ \Phi_{(m,\underline{0}),(n,\underline{0})}^{(n-m,\underline{0}),1} \circ \Phi_{\rm USS}^{m} \, .
    \end{align}
    The map $\Phi_{\rm cl}^{m,n}$ is a quantum channel only if the input is restricted to $(\mathbb{C}^{d})_{\rm sym}^{\otimes m}$, which is the case for symmetric cloning. Theorem~\ref{thm:complexity_unitary_equivariant_permutation_invariant_quantum channels} now gives us the memory and gate complexity of $\Phi_{\rm cl}^{m,n}$ up to our implementation of $\iota_{(n,\underline{0}),\overline{(n-m,\underline{0})}}^{(m,\underline{0}),1}$. We easily check that there is just one way to remove $n-m$ boxes from $(n,\underline{0})$ to obtain $(m,\underline{0})$. In other words, the vector space $\mathcal{P}_{(n,\underline{0})\rightarrow (m,\underline{0})}^{0,n-m,d}$ is one-dimensional. Proposition \ref{prop:implement_paths_embedding} then immediately gives us the memory and gate complexity of $\iota_{(n,\underline{0}),\overline{(n-m,\underline{0})}}^{(m,\underline{0}),1}$ together with the partial trace.
\end{proof}

Figure \ref{fig:example_path_cloning} contains an example for the paths appearing in the implementation of $\Phi_{\rm cl}^{m,n}$.

\begin{figure}[h]
    \begin{align*}
        \ytableausetup{boxsize=1em}
        \ydiagram{4} \quad \rightarrow \quad \ydiagram{5} \otimes \overline{\ydiagram{1}} \quad \rightarrow \quad ... \quad \rightarrow \quad \ydiagram{8} \otimes \overline{\ydiagram{1}}^{\otimes 4}
        \ytableausetup{boxsize=0.5em}
    \end{align*}
    \caption{
    \label{fig:example_path_cloning}
    Path used to implement $\Phi_{\rm cl}^{4,8}$. We trace out the four copies $\overline{\ydiagram{1}}^{\otimes 4}$ of the dual defining irrep.}
\end{figure}

\subsection{Purity amplification}

Assume that Alice creates a pure quantum state $\ket{\psi} \in \mathbb{C}^{d}$ and wants to transmit it to Bob. However, in the process of creating, sending and receiving the state, there is some noise that replaces the state with a random output $\ket{i}\in\mathbb{C}^{d}$ for $1\leq i \leq d$. Classically, Alice can just send her message multiple times and Bob can, with high probability, deduce the original message. This serves as an inspiration for the \emph{purity amplification} protocol, which allows Alice to transmit multiple noisy copies of $\ket{\psi}$ so that Bob can reconstruct the pure state with high fidelity. This problem has been investigated first for qubits in \cite{Cirac_1999}, and a slight variation of it in the form of quantum majority vote was investigated in \cite{Buhrman_2022}. Recently, the case of qudits has been treated in \cite{Li_2025}. The task of purity amplification can formally be defined as follows.
    
\begin{definition}[Purity amplification]
    Let $0 < \alpha < 1$. The task of \textit{purity amplification} is given as the task of implementing the \textit{purity amplification quantum channel} $\Phi_{\rm pa}^{m,n,\alpha} \in \mathcal{C}^{d}(m,n)$, which is defined as the quantum channel that maximizes the $m \rightarrow n$, $\alpha$-depolarizing fidelity
    \begin{align}
        \mathcal{F}_{\rm pa}(\Phi) := \min\limits_{\ket{\psi}\in \mathbb{C}^{d}} \tr[\ketbra{\psi}{\psi}^{\otimes n}\Phi(\rho_{\psi}^{\otimes m})] \, ,
    \end{align}
    where
    \begin{align}
        \rho_{\psi} := (1-\alpha)\ketbra{\psi}{\psi} + \frac{\alpha}{d} \id_{d} \, .
    \end{align}
\end{definition}

In \cite{Li_2025}, it was shown that for the optimal quantum channel is independent of $\alpha$, and that for $n=1$ it is given as
\begin{align}
    \Phi_{\rm pa}^{m,1,\alpha} = \Phi_{\rm pa}^{m,1} = \left(\bigoplus_{\lambda \vdash_{d} m} \Phi_{\lambda, \, \ydiagram{1}}^{\overline{\gamma_{\lambda,\min}},1}\right) \circ \Phi_{\rm USS}^{m} \, ,
\end{align}
where $\gamma_{\lambda,\min} \in \lambda +_{d} \overline{\ydiagram{1}}$ is the unique partition so that for all $\gamma \in \lambda +_{d} \overline{\ydiagram{1}}$ we have $\gamma_{\lambda,min} \preceq \gamma$. In practice, this means taking $\lambda = (\lambda_{1}, \lambda_{2}, ..., \lambda_{d})$ and finding the smallest $1 \leq i \leq d$ so that $\lambda_{i} > \lambda_{i+1}$ (with the convention that $\lambda_{d}>\lambda_{d+1}$) and then take $\gamma_{\lambda,min} = (\lambda_{1}, ..., \lambda_{i}-1, ..., \lambda_{d})$. Again, we apply the algorithms developed in this work to the above result and get the following corollary.

\begin{corollary}
    \label{cor:result_purity_amplification}
    We implement the operation $\rho \rightarrow \Phi_{\rm pa}^{m,1}(\rho)$ in a streaming manner up to error $\epsilon$ in diamond norm with memory complexity $M$ and gate complexity $T$, where
    \begin{align}
        &M=O\big(d^{2}\log_{2}^{p}(d,m,1/\epsilon)\big) \, ,\\
        \label{equ:gate_complexity_reduced}
        &T=O\big(md^{4}\log_{2}^{p}(d,m,1/\epsilon)\big) \, .
    \end{align}
    Here, $p\approx 1.44$.
\end{corollary}

\begin{proof}
    Since $\mathcal{Q}_{\ydiagram{1}}^{d} = \mathbb{C}^{d}$, we see that $\iota_{\ydiagram{1},\gamma_{\lambda,\min}}^{\lambda}$ can be implemented via a path in $\mathcal{P}_{\gamma_{\lambda,\min} \rightarrow \lambda}^{1,0,d}$, as described in Proposition \ref{prop:implement_paths_embedding}. It is easy to see that there is only one way to add a specific box, so $\mathcal{P}_{\gamma_{\lambda,\min}\rightarrow \lambda}^{1,0,d}$ is one-dimensional. Therefore, we ignore the memory and gate complexity for dealing with the path. Using Theorem \ref{thm:complexity_unitary_equivariant_permutation_invariant_quantum channels} together with Proposition \ref{prop:implement_paths_embedding} then gives us the desired result.
\end{proof}

Figure \ref{fig:example_path_purity_amplification} contains an example for the paths appearing in the implementation of $\Phi_{\rm pa}^{m,1}$.

\begin{figure}[h]
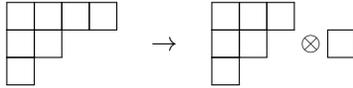

    \begin{align*}
        \ytableausetup{boxsize=1em}
        \ydiagram{4,2,1} \quad \rightarrow \quad \ydiagram{3,2,1} \otimes \ydiagram{1}
        \ytableausetup{boxsize=0.5em}
    \end{align*}
    \caption{    
    \label{fig:example_path_purity_amplification}
    Path used to implement $\Phi_{\rm pa}^{7,1}$ for a given $\mu=(4,2,1)$. We trace out the irrep labelled by $(3,2,1)$ in the end.}
\end{figure}

\section*{Acknowledgements}
We are grateful to Marco Fanizza for suggesting the state symmetrization application, as well as for various insightful comments. We are further grateful for helpful discussion with Tommaso Aschieri, Thomas Fraser, Dmitry Grinko, Zhaoyi Li, Harold Nieuwboer and Māris Ozols.

ET and LM are supported by ERC grant (QInteract, Grant No 101078107) and VILLUM FONDEN (Grant No 10059 and 37532).

\bibliographystyle{IEEEtran}
\bibliography{references}

\appendix

\section*{Appendix}

\subsection{Unitary-equivariant quantum channels on irreps}

In this section, we deal with vector space made up of tensors and different ways of interpreting these tensors as operators on other vector space. To facilitate some calculations, we assume that we have chosen an orthonormal basis $\{\ket{i}\}_{i\in I}$ on each of the vector space we're dealing with. It is important to realize however, that the results in this section are true for all bases, with the exception of those statements that explicitly include a transpose or vectorization.

\begin{definition}
    Let $M\in\mathcal{B}(\mathbb{C}^{a},\mathbb{C}^{b})$, and let $M$ be given by
    \begin{align}
        \sum_{j=1}^{b}\sum_{j=1}^{a} M_{ij} \ketbra{i}{j} \, .
    \end{align}
    Then we define the \emph{vectorization} of $M$ as
    \begin{align}
        \vec (M) := \sum_{j=1}^{b}\sum_{j=1}^{a} M_{ij} \ket{i} \otimes \ket{j} \quad \in \mathbb{C}^{b} \otimes \mathbb{C}^{a} \, .
    \end{align}
\end{definition}

\begin{remark}
    The vectorization is just the partial transpose on the input space. In addition, for $M\in\mathcal{B}(\mathbb{C}^{a},\mathbb{C}^{b})$, we also interpret $\vec(M) \in \mathbb{C}^{b} \otimes \mathbb{C}^{a} \stackrel{SU(d)}{\cong} \mathbb{C}^{a} \otimes \mathbb{C}^{b}$, as there is a canonical isomorphism between the two spaces.
\end{remark}

\begin{fact}
    \label{fct:trace_vectorization}
    Let $M\in\mathcal{B}(\mathbb{C}^{a},\mathbb{C}^{b})$ and let $A\in \mathcal{B}(\mathbb{C}^{a})$ and $B\in \mathcal{B}(\mathbb{C}^{b})$. Let further $\ket{V}:=\vec(M)$. Then we have
    \begin{align}
        \tr_{\mathbb{C}^{a}}\left[\ketbra{V}{V}(A^{T} \otimes B)\right] = MAM^{*}B \, .
    \end{align}
\end{fact}

We turn to the Littlewood-Richardson coefficients and their associated multiplicity spaces.

\begin{proposition}
    \label{prop:multiplicity_irreps}    
    Let $\lambda\vdash_{d}m$, $\mu\vdash_{d}n$ and $\nu\vdash_{d}l$. Then we have
    \begin{align}
        c_{\lambda,\mu}^{\nu} = c_{\mu,\lambda}^{\nu}= c_{\mu,\overline{\nu}}^{\overline{\lambda}} = c_{\overline{\lambda},\overline{\mu}}^{\overline{\nu}}\, .
    \end{align}
    Let now $c$ denote the integer given by the equality above. There is a consistent choice of $U_{\rm CG}^{\lambda,\mu,d}$, $U_{\rm CG}^{\mu,\overline{\nu},d}$ and $U_{\rm CG}^{\overline{\lambda},\overline{\mu},d}$ so that we have
    \begin{align}
        \vec \left((U_{\rm CG}^{\lambda,\mu,d})^{*}|_{\mathcal{Q}_{\nu}^{d}\otimes \mathbb{C}^{c}}\right) &= \vec \left((U_{\rm CG}^{\mu,\lambda,d})^{*}|_{\mathcal{Q}_{\nu}^{d}\otimes \mathbb{C}^{c}}\right) = \\
        &= \sqrt{\frac{\dim \mathcal{Q}_{\nu}^{d}}{\dim \mathcal{Q}_{\lambda}^{d}}} \vec \left((U_{\rm CG}^{\mu,\overline{\nu},d})^{*}|_{\overline{\mathcal{Q}_{\lambda}^{d}} \otimes \mathbb{C}^{c}}\right) = \\
        &= \vec \left((\id_{\overline{\mathcal{Q}_{\nu}^{d}}}\otimes \id_{\mathbb{C}^{c}})U_{\rm CG}^{\overline{\lambda},\overline{\mu},d}\right) \, ,
    \end{align}
    where the image of the last operator is restricted to $\overline{\mathcal{Q}_{\nu}^{d}} \otimes \mathbb{C}^{c}$.
\end{proposition}

\begin{remark}
    The Clebsch--Gordan transforms are uniquely defined up to a phase and up to their image on the multiplicity spaces. Proposition \ref{prop:multiplicity_irreps} now tells us that there is a way to compare multiplicity spaces for corresponding triples $_{\lambda,{\mu}}^{\nu}$ between different Clebsch--Gordan transforms. Notice also that if $c_{\lambda,\mu}^{\nu}=1$, there is a unique choice on the multiplicity spaces and thus the operators are uniquely defined up to phase.
\end{remark}

\begin{remark}
    We apply Proposition \ref{prop:multiplicity_irreps} multiple times to compare the triples
    \begin{align}
        _{\lambda,{\mu}}^{\nu} \quad \leftrightarrow \quad _{\lambda,\overline{\nu}}^{\overline{\mu}} \quad \leftrightarrow \quad _{\overline{\lambda},\nu}^{\mu} \, ,
    \end{align}
    which gives us
    \begin{align}
        \vec \left((U_{\rm CG}^{\lambda,\mu,d})^{*}|_{\mathcal{Q}_{\nu}^{d}\otimes \mathbb{C}^{c}}\right) = \sqrt{\frac{\dim \mathcal{Q}_{\nu}^{d}}{\dim \mathcal{Q}_{\lambda}^{d}}} \vec \left((\id_{\overline{\mathcal{Q}_{\mu}^{d}}}\otimes \id_{\overline{\mathbb{C}}^{c}})U_{\rm CG}^{\overline{\lambda},\nu,d}\right) \, .
    \end{align}
\end{remark}

\begin{proof}
    The equality $c_{\lambda,\mu}^{\nu}= c_{\mu,\lambda}^{\nu}$ and the corresponding equality of vectorizations follow immediately from the commutativity of the tensor product. For the next equality, we use Schur's Lemma and see that the number of linearly independent homomorphisms from $\mathcal{Q}_{\nu}^{d}$ to $\mathcal{Q}_{\lambda}^{d} \otimes \mathcal{Q}_{\mu}^{d}$ that commute with $SU(d)$ is exactly the multiplicity between the relevant irreps, that is
    \begin{align}
        \label{equ:proof_dimension_homomorphisms}
        \dim \mathcal{B}_{SU(d)} \left(\mathcal{Q}_{\nu}^{d},\mathcal{Q}_{\lambda}^{d}\otimes \mathcal{Q}_{\mu}^{d}\right) = c_{\lambda,\mu}^{\nu} \, .
    \end{align}
    Here, we denote by $\mathcal{B}_{SU(d)}(V,W)$ the elements of $\mathcal{B}(V,W)$ that commute with the action of $SU(d)$. By vectorization, we interpret
    \begin{align}
        \vec \left(\mathcal{B}_{SU(d)} \left(\mathcal{Q}_{\nu}^{d},\mathcal{Q}_{\lambda}^{d}\otimes \mathcal{Q}_{\mu}^{d}\right) \right) \subseteq \overline{\mathcal{Q}_{\nu}^{d}} \otimes \mathcal{Q}_{\lambda}^{d} \otimes \mathcal{Q}_{\mu}^{d} \, .
    \end{align}
    In particular, $\vec \left(\mathcal{B}_{SU(d)} \left(\mathcal{Q}_{\nu}^{d},\mathcal{Q}_{\lambda}^{d}\otimes \mathcal{Q}_{\mu}^{d}\right) \right)$ is equal to the subspace of trivial $SU(d)$ irreps, by the requirement that the homomorphisms commute with the action of $SU(d)$. The trivial $SU(d)$ irrep is one-dimensional, so by comparing dimensions with Equation \eqref{equ:proof_dimension_homomorphisms}, we see that there are $c_{\lambda,\mu}^{\nu}$ copies of it in $\overline{\mathcal{Q}_{\nu}^{d}} \otimes \mathcal{Q}_{\lambda}^{d} \otimes \mathcal{Q}_{\mu}^{d}$. We take $\mathcal{B}_{SU(d)} \left(\overline{\mathcal{Q}_{\lambda}^{d}}, \mathcal{Q}_{\mu}^{d} \otimes \overline{\mathcal{Q}_{\nu}^{d}}\right)$ and we see that it is also isomorphic to the invariant subspace of $\overline{\mathcal{Q}_{\nu}^{d}} \otimes \mathcal{Q}_{\lambda}^{d} \otimes \mathcal{Q}_{\mu}^{d}$. By the same argument as before, we get that $c_{\lambda,\mu}^{\nu}= c_{\mu,\overline{\nu}}^{\overline{\lambda}}$. Finally, we know that the multiplicity of a given irrep in a vector space $\mathcal{H}$ is the same as the multiplicity of the dual irrep under the dual action on $\overline{\mathcal{H}}$. Since the trivial irrep is self-dual, it has multiplicity $c_{\mu,\lambda}^{\nu}$ in $\mathcal{Q}_{\nu}^{d} \otimes \overline{\mathcal{Q}_{\lambda}^{d}} \otimes \overline{\mathcal{Q}_{\mu}^{d}}$. Going through the previous argument with $\mathcal{B}_{SU(d)} \left(\overline{\mathcal{Q}_{\nu}^{d}}, \overline{\mathcal{Q}_{\lambda}^{d}} \otimes \overline{\mathcal{Q}_{\mu}^{d}}\right)$ then yields $c_{\lambda,\mu}^{\nu} = c_{\overline{\lambda},\overline{\mu}}^{\overline{\nu}}$.

    The previous arguments allow us to set $c:=c_{\mu,\lambda}^{\nu}= c_{\mu,\overline{\nu}}^{\overline{\lambda}} = c_{\overline{\lambda},\overline{\mu}}^{\overline{\nu}}$. For the next statement, we consider the restriction $(U_{\rm CG}^{\lambda,\mu,d})^{*}|_{\mathcal{Q}_{\nu}^{d}\otimes \mathbb{C}^{c}}$, and we reinterpret its vectorization as follows
    \begin{align}
        &A \in \mathcal{B}_{SU(d)} \left(\overline{\mathcal{Q}_{\lambda}^{d}} \otimes \mathbb{C}^{c}, \mathcal{Q}_{\mu}^{d} \otimes \overline{\mathcal{Q}_{\nu}^{d}}\right) \, , \\
        &\vec (A) := \ket{V} = \vec \left((U_{\rm CG}^{\lambda,\mu,d})^{*}|_{\mathcal{Q}_{\nu}^{d}\otimes \mathbb{C}^{c}}\right) \in \mathcal{Q}_{\lambda}^{d} \otimes \mathcal{Q}_{\mu}^{d} \otimes \overline{\mathcal{Q}_{\nu}^{d}} \otimes \mathbb{C}^{c} \, .
    \end{align}
    We know that $\ket{V}$ lies in the subspace of trivial irreps, so by Schur's lemma we have
    \begin{align}
        \label{equ:proof_multiplicity_irreps_A*A}
        A^{*}A = \frac{1}{\dim \overline{\mathcal{Q}_{\lambda}^{d}}} \id_{\overline{\mathcal{Q}_{\lambda}^{d}}} \otimes \tr_{\overline{\mathcal{Q}_{\lambda}^{d}}}\left[A^{*}A\right] \, .
    \end{align}
    We further deduce
    \begin{align}
        \tr_{\overline{\mathcal{Q}_{\lambda}^{d}}}\left[A^{*}A\right] = \tr_{\mathcal{Q}_{\nu}^{d}}\left[\left((U_{\rm CG}^{\lambda,\mu,d})^{*}|_{\mathcal{Q}_{\nu}^{d}\otimes \mathbb{C}^{c}}\right)^{*}(U_{\rm CG}^{\lambda,\mu,d})^{*}|_{\mathcal{Q}_{\nu}^{d}\otimes \mathbb{C}^{c}}\right] = \dim \mathcal{Q}_{\nu}^{d} \, \id_{\mathbb{C}^{c}} \, .
    \end{align}
    For the last equality we have used that $U_{\rm CG}^{\lambda,\mu,d}$ is unitary. Together with Equation \eqref{equ:proof_multiplicity_irreps_A*A} we see that $\sqrt{\frac{\dim \overline{\mathcal{Q}_{\lambda}^{d}}}{\dim \mathcal{Q}_{\nu}^{d}}}A$ is an isometry and commutes with the action of $SU(d)$. This means it is equal to $(U_{\rm CG}^{\mu,\overline{\nu},d})^{*}|_{\overline{\mathcal{Q}_{\lambda}^{d}} \otimes \mathbb{C}^{c}}$ up to a unitary rotation $U_{c}$ on $\mathbb{C}^{c}$. For a given choice of $U_{\rm CG}^{\lambda,\mu,d}$, we simply choose $U_{\rm CG}^{\mu,\overline{\nu},d}$ so that $U_{c}=\id_{\mathbb{C}^{c}}$, which gives
    \begin{align}
        \vec \left((U_{\rm CG}^{\lambda,\mu,d})^{*}|_{\mathcal{Q}_{\nu}^{d}\otimes \mathbb{C}^{c}}\right) = \sqrt{\frac{\dim \mathcal{Q}_{\nu}^{d}}{\dim \overline{\mathcal{Q}_{\lambda}^{d}}}} \vec \left((U_{\rm CG}^{\mu,\overline{\nu},d})^{*}|_{\overline{\mathcal{Q}_{\lambda}^{d}} \otimes \mathbb{C}^{c}}\right) \, .
    \end{align}
    Together with the fact that $\dim \overline{\mathcal{Q}_{\lambda}^{d}} = \dim \mathcal{Q}_{\lambda}^{d}$, we get the first equality.
    
    For the last equality, we take
    \begin{align}
        &\vec (B^{*}) := \ket{V} \in \mathcal{Q}_{\lambda}^{d} \otimes \mathcal{Q}_{\mu}^{d} \otimes \overline{\mathcal{Q}_{\nu}^{d}} \otimes \mathbb{C}^{c} \, , \\
        &B \in \mathcal{B}_{SU(d)} \left(\overline{\mathcal{Q}_{\lambda}^{d}} \otimes \overline{\mathcal{Q}_{\mu}^{d}}, \overline{\mathcal{Q}_{\nu}^{d}} \otimes \mathbb{C}^{c}\right) \, .
    \end{align}
    By similar arguments as above, we see that there is a choice for $U_{\rm CG}^{\overline{\lambda},\overline{\mu},d}$ so that
    \begin{align}
        B = (\id_{\overline{\mathcal{Q}_{\nu}^{d}}}\otimes \id_{\mathbb{C}^{c}})U_{\rm CG}^{\overline{\lambda},\overline{\mu},d} \, .
    \end{align}
    We inductively fix an arbitrary basis on the multiplicity spaces for all triples $(\lambda,\mu,\nu)$ and the triples related to them by the calculations above, and thereby obtain a consistent choice of $U_{\rm CG}^{\lambda,\mu,d}$, $U_{\rm CG}^{\mu,\overline{\nu},d}$ and $U_{\rm CG}^{\overline{\lambda},\overline{\mu},d}$ that fulfils the above requirements.
\end{proof}

Let $\mathcal{C}_{u}(\mathcal{Q}_{\lambda}^{d},\mathcal{Q}_{\mu}^{d})$ be the set of unitary-equivariant quantum channels from $\mathcal{Q}_{\lambda}^{d}$ to $\mathcal{Q}_{\mu}^{d}$. The following is the main theorem of this section, which has been stated for $SU(2)$ in \cite{Aschieri_2024}, and which has been partially introduced in \cite{Nuwairan_2013}.

\begin{theorem}
    \label{thm:equivalence_of_unitary_equivariant_quantum channels}
    The extremal points of $\mathcal{C}_{u}(\mathcal{Q}_{\lambda}^{d},\mathcal{Q}_{\mu}^{d})$ are given by the quantum channels $\Phi_{\lambda,\mu}^{\gamma,\psi}$. Further, for $A\in \mathcal{B}(\mathcal{Q}_{\lambda}^{d})$, we find
    \begin{align}
        \label{equ:equivariant_quantum channels_forms_1}
        \Phi_{\lambda,\mu}^{\gamma,\psi}(A) &= \frac{\dim \mathcal{Q}_{\lambda}^{d}}{\dim \mathcal{Q}_{\gamma}^{d}}\tr_{\overline{\mathcal{Q}_{\lambda}^{d}}}\left[ \iota_{\overline{\lambda},\mu}^{\gamma,\psi} (\iota_{\overline{\lambda},\mu}^{\gamma,\psi})^{*} (A^{T}\otimes \id_{\mathcal{Q}_{\mu}^{d}})\right] \, , \\
        \label{equ:equivariant_quantum channels_forms_2}
        \Phi_{\lambda,\mu}^{\gamma,\psi}(A) &= \tr_{\overline{\mathcal{Q}_{\gamma}^{d}}}\left[\iota_{\mu,\overline{\gamma}}^{\lambda,\psi}A(\iota_{\mu,\overline{\gamma}}^{\lambda,\psi})^{*}\right] \, , \\
        \label{equ:equivariant_quantum channels_forms_3}
        \Phi_{\lambda,\mu}^{\gamma,\psi}(A) &= \frac{\dim \mathcal{Q}_{\lambda}^{d}}{\dim \mathcal{Q}_{\mu}^{d}} (\iota_{\lambda,\gamma}^{\mu,\overline{\psi}})^{*}(A\otimes \id_{\mathcal{Q}_{\gamma}^{d}}) \iota_{\lambda,\gamma}^{\mu,\overline{\psi}} \, .
    \end{align}
    Here, the vector $\ket{\overline{\psi}}$ is the complex conjugate of $\ket{\psi}$ in the same bases used to define the transpose for $\mathbb{C}^{c_{\overline{\lambda},\mu}^{\gamma}}$.
\end{theorem}

\begin{proof}
    For the first part, we proceed similarly to the proof of Theorem \ref{thm:classification_Cusd}. Let $\Phi\in \mathcal{C}_{u}(\mathcal{Q}_{\lambda}^{d},\mathcal{Q}_{\mu}^{d})$. Then the Choi matrix $C_{\Phi}$ is given by
    \begin{align}
        C_{\Phi} = (U_{\rm CG}^{\overline{\lambda},\mu,d})^{*} \left(\bigoplus_{\gamma \in \overline{\lambda} +_{d} \mu} \frac{\dim \mathcal{Q}_{\lambda}^{d}}{\dim \mathcal{Q}_{\gamma}^{d}} \id_{\mathcal{Q}_{\gamma}^{d}} \otimes M_{\gamma} \right) U_{\rm CG}^{\overline{\lambda},\mu,d} \, ,
    \end{align}
    with $M_{\gamma} \in \mathcal{B}(\mathbb{C}^{c_{\overline{\lambda},\mu}^{\gamma}})$ and
    \begin{align}
        M_{\gamma} \geq 0 \quad , \quad \sum_{\gamma \in \overline{\lambda} +_{d} \mu} \tr M_{\gamma} = 1 \, .
    \end{align}
    Any extremal quantum channel must have $\tr M_{\gamma} = 1$ for some $\gamma$, and $M_{\gamma} = 0$ for all $\gamma^{'} \neq \gamma$. This means we interpret the $M_{\gamma}$ as a state. It is well known that the extremal points of the set of states are rank $1$ projections, and therefore $M_{\gamma} = \ketbra{\psi}{\psi}$ for some $\ket{\psi} \in \mathbb{C}^{c_{\overline{\lambda},\mu}^{\gamma}}$. This yields exactly the definition of $\Phi_{\lambda,\mu}^{\gamma,\psi}$ given in Definition \ref{def:unitary_equivariant_quantum channel_irreps}.

    For the second part, we remember the definition of $\Phi_{\lambda,\mu}^{\gamma,\psi}(A)$ via its Choi matrix as
    \begin{align}
        \Phi_{\lambda,\mu}^{\gamma,\psi}(A) := \frac{\dim \mathcal{Q}_{\lambda}^{d}}{\dim \mathcal{Q}_{\gamma}^{d}}\tr_{\overline{\mathcal{Q}_{\lambda}^{d}}}\left[ (U_{\rm CG}^{\overline{\lambda},\mu,d})^{*}(\id_{\mathcal{Q}_{\gamma}^{d}}\otimes \ketbra{\psi}{\psi})U_{\rm CG}^{\overline{\lambda},\mu,d} (A^{T}\otimes \id_{\mathcal{Q}_{\mu}^{d}})\right] \, ,
    \end{align}
    and together with the definition of $\iota_{\overline{\lambda},\mu}^{\gamma,\psi}$ as
    \begin{align}
        \iota_{\overline{\lambda},\mu}^{\gamma,\psi}(\ket{q_{\gamma}^{d}}) := (U_{\rm CG}^{\overline{\lambda},\mu,d})^{*}(\ket{q_{\gamma}^{d}} \otimes \ket{\psi})
    \end{align}
    we obtain Equation \eqref{equ:equivariant_quantum channels_forms_1}. By applying Proposition \ref{prop:multiplicity_irreps} to get an equivalency between the triples
    \begin{align}
        (\overline{\lambda},\mu,\gamma) \quad , \quad  (\mu,\overline{\gamma},\lambda) \quad , \quad (\overline{\lambda},\overline{\gamma},\overline{\mu}) \, ,
    \end{align}
    we see that $c:=c_{\mu,\lambda}^{\nu}= c_{\mu,\overline{\nu}}^{\overline{\lambda}} = c_{\lambda,\gamma}^{\mu}$, and we define
    \begin{align}
        \label{equ:proof_definition_vectorized_CG_transform}
        \ket{V} := \vec \left((U_{\rm CG}^{\overline{\lambda},\mu,d})^{*}|_{\mathcal{Q}_{\gamma}^{d}\otimes \mathbb{C}^{c}}\right) &= \sqrt{\frac{\dim \mathcal{Q}_{\gamma}^{d}}{\dim \mathcal{Q}_{\lambda}^{d}}} \vec \left((U_{\rm CG}^{\mu,\overline{\gamma},d})^{*}|_{\mathcal{Q}_{\lambda}^{d} \otimes \mathbb{C}^{c}}\right) = \\
        &= \sqrt{\frac{\dim \mathcal{Q}_{\gamma}^{d}}{\dim \mathcal{Q}_{\mu}^{d}}} \vec \left((\id_{\mathcal{Q}_{\mu}^{d}}\otimes \id_{\mathbb{C}^{c}})U_{\rm CG}^{\lambda,\gamma,d}\right) \, .
    \end{align}
    We further define the operator $C_{\overline{\gamma},\overline{c},\overline{\lambda},\mu}$ for $C_{\gamma}\in \mathcal{B}(\mathcal{Q}_{\gamma}^{d})$, $C_{\lambda}\in \mathcal{B}(\mathcal{Q}_{\lambda}^{d})$, $C_{\mu}\in \mathcal{B}(\mathcal{Q}_{\mu}^{d})$ and $C_{c}\in \mathcal{B}(\mathbb{C}^{c})$ as
    \begin{align}
        C_{\overline{\gamma},\overline{c},\overline{\lambda},\mu} := C_{\gamma}^{T}\otimes C_{c}^{T} \otimes C_{\lambda}^{T} \otimes C_{\mu} \, .
    \end{align}
    The definition of $\ket{V}$ in Equation \eqref{equ:proof_definition_vectorized_CG_transform} together with Fact \ref{fct:trace_vectorization} then give us
    \begin{align}
        &\tr_{\overline{\mathcal{Q}_{\gamma}^{d}} \otimes \overline{\mathbb{C}^{c}}} \left[\ketbra{V}{V} C_{\overline{\gamma},\overline{c},\overline{\lambda},\mu}\right] = (U_{\rm CG}^{\overline{\lambda},\mu,d})^{*}(C_{\gamma}\otimes C_{c})U_{\rm CG}^{\overline{\lambda},\mu,d} (C_{\lambda}^{T} \otimes C_{\mu}) \, , \\
        &\tr_{\overline{\mathcal{Q}_{\gamma}^{d}} \otimes \overline{\mathcal{Q}_{\lambda}^{d}}} \left[\ketbra{V}{V} C_{\overline{\gamma},\overline{c},\overline{\lambda},\mu}\right] = \frac{\dim \mathcal{Q}_{\gamma}^{d}}{\dim \mathcal{Q}_{\lambda}^{d}} (U_{\rm CG}^{\overline{\gamma},\mu,d})^{*}(C_{\lambda}\otimes C_{c})U_{\rm CG}^{\overline{\gamma},\mu,d} ( C_{\gamma}^{T} \otimes C_{\mu}) \, , \\
        &\tr_{\overline{\mathcal{Q}_{\lambda}^{d}} \otimes \overline{\mathbb{C}^{c}}} \left[\ketbra{V}{V} C_{\overline{\gamma},\overline{c},\overline{\lambda},\mu}\right] = \frac{\dim \mathcal{Q}_{\gamma}^{d}}{\dim \mathcal{Q}_{\mu}^{d}} U_{\rm CG}^{\lambda,\gamma,d} (C_{\lambda} \otimes C_{\gamma}) (U_{\rm CG}^{\lambda,\gamma,d})^{*}(C_{\mu}\otimes C_{c}^{T}) \, .
    \end{align}
    Taking the trace over the remaining subspace apart from $\mathcal{Q}_{\mu}^{d}$ and setting $C_{\gamma}=\id_{\mathcal{Q}_{\gamma}^{d}}$, $C_{\lambda}=A$, $C_{\mu}=\id_{\mathcal{Q}_{\mu}^{d}}$ and $C_{c}=\ketbra{\psi}{\psi}$, we get
    \begin{align}
        \Phi_{\lambda,\mu}^{\gamma,\psi}(A) &= \frac{\dim \mathcal{Q}_{\lambda}^{d}}{\dim \mathcal{Q}_{\gamma}^{d}} \tr_{\overline{\mathcal{Q}_{\lambda}^{d}}}\left[(U_{\rm CG}^{\overline{\lambda},\mu,d})^{*}(\id_{\mathcal{Q}_{\gamma}^{d}}\otimes \ketbra{\psi}{\psi})U_{\rm CG}^{\overline{\lambda},\mu,d}(A^{T}\otimes \id_{\mathcal{Q}_{\mu}^{d}})\right]= \\
        &= \tr_{\overline{\mathcal{Q}_{\gamma}^{d}}}\left[(U_{\rm CG}^{\overline{\gamma},\mu,d})^{*}(A\otimes \ketbra{\psi}{\psi})U_{\rm CG}^{\overline{\gamma},\mu,d}\right] = \\
        &= \frac{\dim \mathcal{Q}_{\lambda}^{d}}{\dim \mathcal{Q}_{\mu}^{d}} \tr_{\mathbb{C}^{c}}\left[U_{\rm CG}^{\lambda,\gamma,d}(A \otimes \id_{\mathcal{Q}_{\gamma}^{d}}) (U_{\rm CG}^{\lambda,\gamma,d})^{*}(\id_{\mathcal{Q}_{\mu}^{d}}\otimes (\ketbra{\psi}{\psi})^{T})\right] \, .
    \end{align}
    The first two equalities are already of the desired form. For the third equality, we just need to remember that $(\ketbra{\psi}{\psi})^{T} = \ketbra{\overline{\psi}}{\overline{\psi}}$, where the complex conjugation is performed with respect the same bases used to define the transpose.
\end{proof}

\subsection{Uniformly random Gel'fand-Tsetlin base vector for $\mathcal{P}_{\lambda}^{d}$}

We introduce two algorithms that allow us to sample uniformly at random from the Gel'fand-Tsetlin base vector for $\mathcal{P}_{\lambda}^{d}$. Algorithm \ref{alg:sample_one_box_removed} lets us obtain a partition of $m-1$ from a partition of $m$. Example \ref{exa:algorithm_sample_one_box_removed} provides an instance of this algorithm. Algorithm \ref{alg:sample_GT_basis_vector} consists of a repeated application of Algorithm \ref{alg:sample_one_box_removed}. As proven in \cite{Greene_1979}, we use this algorithm to sample a uniformly random path of partitions corresponding to a Gel'fand-Tsetlin base vector for $\mathcal{P}_{\lambda}^{d}$.

\begin{figure}[h]
\begin{algorithm}[H]\label{alg:sample_one_box_removed}
\SetAlgoLined
\caption{Sample $\mu\in \lambda +_{d} \overline{\ydiagram{1}}$}
\SetKwInOut{Input}{Input}
\SetKwInOut{Output}{Output}
\SetKwRepeat{Repeat}{repeat}{until}
\Input{A partition $\lambda\vdash_{d} m$}
\Output{A partition $\mu\in \lambda +_{d} \overline{\ydiagram{1}}$}

Pick uniformly at random a box in the Young diagram of shape $\lambda$, labelled by column and row indices $(i,j)$\;

\While{there is a box to the right or below $(i,j)$}{
    Pick uniformly at random a box $(i',j')$ with either $i'=i$, $j'>j$ or $i'>i$\;
    Set $(i,j)\leftarrow (i',j')$\;
}

Set $\mu \leftarrow \lambda$\;
Set $\mu^{i}\leftarrow \mu^{i} - 1$\;

\Return $\mu$
\end{algorithm}
\end{figure}

\begin{example}
    \label{exa:algorithm_sample_one_box_removed}
    We assume that $\lambda = (5,3,3,2)$. Then a possible sequence of $(i,j)$ is given by
    \begin{align}
        (1,2) \rightarrow (3,2) \rightarrow (3,3) \, .
    \end{align}
    We visualize this as
    \begin{align}
        \ytableausetup{boxsize=1em}
        \ytableaushort{\none} * [*(yellow)]{5,3,3,2} \quad \rightarrow \quad \ytableaushort{\none \cross} * [*(yellow)]{2+3,1+1,1+1,1+1} * {5,3,3,2} \quad \rightarrow \quad \ytableaushort{\none, \none, \none \cross} * [*(yellow)]{0,0,2+1,1+1} * {5,3,3,2} \quad \rightarrow \quad \ytableaushort{\none, \none, \none \none \cross} * {5,3,3,2} \, .
        \ytableausetup{boxsize=0.5em}
    \end{align}
    The coloured part in each step indicates the boxes from which we pick uniformly at random and the $\cross$ indicates the particular instance of our choice. For this sequence, we remove the box $(3,3)$ to obtain $\mu = (5,3,2,2)$.    
\end{example}

\begin{figure}[h]
\begin{algorithm}[H]\label{alg:sample_GT_basis_vector}
\SetAlgoLined
\caption{Uniformly random sampling of a Gel'fand-Tsetlin base vector}
\SetKwInOut{Input}{Input}
\SetKwInOut{Output}{Output}
\SetKwRepeat{Repeat}{repeat}{until}
\Input{A partition $\lambda\vdash_{d} m$}
\Output{A sequence of partitions $p_{\lambda}^{m}$ corresponding to a Gel'fand-Tsetlin base vector for $\mathcal{P}_{\lambda}^{d}$}

Set $\lambda_{m}\leftarrow\lambda$\;

\For{$i=d$ to $1$}{
    Apply Algorithm \ref{alg:sample_one_box_removed} to $\lambda_{i}$ to obtain $\lambda_{i-1}$\;
}

\Return $p_{\lambda}=(\lambda_{0},...,\lambda_{m})$
\end{algorithm}
\end{figure}

\begin{theorem}[Section 2 of \cite{Greene_1979}]
    \label{thm:algorithm_sample_GT_basis}
    Let $\lambda \vdash_{d} m$ with $l(\lambda)=r$. Then we classically uniformly random sample a Gel'fand-Tsetlin base vector for $\mathcal{P}_{\lambda}^{d}$ by applying Algorithm \ref{alg:sample_GT_basis_vector}.
\end{theorem}

The total gate complexity of the algorithm is given by $T=\sum_{i=1}^{m}T_{i}$, where $T_{i}$ is the time necessary to perform Algorithm \ref{alg:sample_one_box_removed} with input $\lambda_{i}$. For a straightforward implementation, we see that $T_{i}=O(i)$, which leads to a total gate complexity of $O(m^{2})$. We reduce this number by applying Algorithm \ref{alg:sample_one_box_removed_alternative} instead, which we show is equivalent to Algorithm \ref{alg:sample_one_box_removed}. Example \ref{exa:algorithm_sample_one_box_removed_alternative} provides an instance of the algorithm. 

\begin{figure}[h]
\begin{algorithm}[H]\label{alg:sample_one_box_removed_alternative}
\SetAlgoLined
\caption{Sample $\mu\in \lambda +_{d} \overline{\ydiagram{1}}$}
\SetKwInOut{Input}{Input}
\SetKwInOut{Output}{Output}
\SetKwRepeat{Repeat}{repeat}{until}
\Input{A partition $\lambda\vdash_{d} m$ with $l(\lambda)=r$}
\Output{A partition $\mu\in \lambda +_{d} \overline{\ydiagram{1}}$}
\# Create the weighted partition $\nu$\;
Define $(k_{1},...,i_{r'})$ as the integers so that $i_{1}=1$, $0<\lambda^{i_{k+1}}$ and $i_{k+1}=:\min\{i:\lambda^{i}<\lambda^{k}\}$ \;
Define the partition $\nu$ with $\nu^{k}:=\lambda^{i_{k}}$ for $1\leq k \leq r'$\;
Define the weights $v(k):=i_{k+1}-i_{k}$ for $1\leq k \leq r'-1$\;
Define the weight $v(r'):=r-i_{r'}$\;
Define the weights $w(l):=\nu^{r'+1-l}-\nu^{r'-l+2}$ for $2\leq l \leq r'$\;
Define the weight $w(1):=\nu^{r'}$\;
\# Pick a starting point\;
Define the probabilities $P(k,l):=v(k)w(l)/(\sum_{a,b}v(a)w(b))$\;
Pick at random a box in the Young diagram of shape $\nu$, labelled by column and row indices $(k,l)$, according to the probability $P(k,l)$\;
\#Pick boxes in the same row or column\;
\While{there is a box to the right or below $(k,l)$}{
    Define the probabilities $Q(k,l'):=w(l')/(\sum_{a>k}v(a)+\sum_{b>j}w(b))$\;
    Define the probabilities $Q(k',l):=v(k')/(\sum_{a>k}v(a)+\sum_{b>j}w(b))$\;
    Pick at random a box $(k',l')$ with either $k'=k$, $l'>l$ or $k'>k$, $l'=l$, according to the probabilities $Q(k',l')$\;
    Set $(k,l)\leftarrow (k',l')$\;
}
\# Return new diagram\;
Set $\mu \leftarrow \lambda$\;
Set $\mu^{i_{k}+v(k)-1}\leftarrow \mu^{i_{k}+v(k)-1} - 1$\;

\Return $\mu$
\end{algorithm}
\end{figure}

\begin{example}
    \label{exa:algorithm_sample_one_box_removed_alternative}
    We assume again $\lambda = (5,3,3,2)$. Then first, we convert $\lambda$ into the weighted tableau given by 
    \begin{align}
        \label{equ:squashing_boxes_YT}
        \ytableausetup{boxsize=1.6em}
        \begin{ytableau}
            1/1 & 1/1 & 1/1 & 1/1 & 1/1 \\
            *(yellow) 1/1 & *(yellow) 1/1 & *(yellow) 1/1 \\
            *(yellow) 1/1 & *(yellow) 1/1 & *(yellow) 1/1 \\
            1/1 & 1/1
        \end{ytableau}
        \quad \rightarrow \quad
        \begin{ytableau}
            *(yellow) 1/1 & *(yellow) 1/1 & 1/1 & *(yellow) 1/1 & *(yellow) 1/1 \\
            *(yellow) 2/1 & *(yellow) 2/1 & 2/1 \\
            *(yellow) 1/1 & *(yellow) 1/1
        \end{ytableau}
        \quad \rightarrow \quad
        \begin{ytableau}
            1/2 & 1/1 & 1/2 \\
            2/2 & 2/1 \\
            1/2
        \end{ytableau} \, .
        \ytableausetup{boxsize=0.5em}
    \end{align}
    The diagram gets squashed from top to bottom and left to right so that there are no duplicate rows or columns. In each step, the colored boxes are being squashed, and the weights $\ytableausetup{boxsize=1.6em} \begin{ytableau} v/w \end{ytableau} \ytableausetup{boxsize=0.5em}$ count the number of squashed rows and columns respectively. A possible sequence of $(k,l)$ is now given by
    \begin{align}
        (1,1) \rightarrow (2,1) \rightarrow (2,2) \, .
    \end{align}
    This sequence corresponds to the sequence given in Example \ref{exa:algorithm_sample_one_box_removed}, and we visualize it as
    \begin{align}
        \ytableausetup{boxsize=1em}
        \ytableaushort{212, 42, 2} * [*(yellow)]{3,2,1} \quad \rightarrow \quad \ytableaushort{\cross 12 , 2, 1} * [*(yellow)]{1+2,1,1} * {3,2,1} \quad \rightarrow \quad \ytableaushort{\none, \cross 1, 1} * [*(yellow)]{3+0,1+1,1} * {3,2,1} \quad \rightarrow \quad \ytableaushort{\none, \none \cross, \none} * {3,2,1} \, .
        \ytableausetup{boxsize=0.5em}
    \end{align}
    The coloured part in each step indicates the boxes from which we pick at random and the $\cross$ indicates the particular instance of our choice. The weights correspond to the relative probability weights in each step. For this sequence, we land on the box $(2,2)$, which means removing the box on the bottom right corner of the rectangle that got squished into $(2,2)$. In our case, this is box $(3,3)$ of $\lambda$, which yields $\mu = (5,3,2,2)$.
\end{example}

\begin{proposition}
    \label{prop:complexity_sample_one_box_removed_alternative}
    Let $\lambda \vdash_{d} m$ with $l(\lambda)=r$. Then Algorithms \ref{alg:sample_one_box_removed} and \ref{alg:sample_one_box_removed_alternative} are isomorphic, and we implement Algorithm \ref{alg:sample_one_box_removed_alternative} using $O(r)$ uniformly random samples on $[0,1]$ and $T=O(r^{2})$ basic arithmetical operations.
\end{proposition}

\begin{proof}
    We first prove the complexity of Algorithm \ref{alg:sample_one_box_removed_alternative}. All the initial definitions take $O(r)$ operations. Further, calculating the probabilities $P(k,l)$ takes $O((r')^{2})$ operations, and we need to sample once for the first step. During the while loop, calculating the probabilities $Q(k,l)$ takes $O(r')$ operations, and we sample a random value in each step. We have at most $r'$ steps, and together with the fact that $r'<r$ this gives us the number of random samples $O(r)$ and the number of arithmetic operations $O(r^{2})$.

    Now we prove that this algorithm is equivalent to Algorithm \ref{alg:sample_one_box_removed}. To this end, we show that the probability of ending in a certain box is the same for boxes that get squashed together, as depicted by Equation \eqref{equ:squashing_boxes_YT}. Let now $n_{r}(i,j)$ be the number of boxes to the right of the box indexed by $(i,j)$ in a given Young diagram, and let $n_{d}(i,j)$ be the number of boxes below. Let further $p_{s}(i,j)$ be the probability to remove the box $(s,\lambda^{s})$ given that Algorithm \ref{alg:sample_one_box_removed} is currently at box $(i,j)$. Then $p_{s}(i,j)$ is given by the sum of the probabilities of jumping to box $(i',j')$ multiplied with $p_{s}(i',j')$. Thus
    \begin{align}
        \label{equ:probability_remove_box_given_location}
        p_{s}(i,j)=\frac{1}{n_{r}(i,j)+n_{d}(i,j)}\left(\sum_{i'=i+1}^{i+n_{d}(i,j)}p_{s}(i',j)+\sum_{j'=j+1}^{j+n_{r}(i,j)}p_{s}(i,j')\right) \, .
    \end{align}
    We prove by induction that if $n_{d}(i,j)=n_{d}(i,j+1)$, then $p_{s}(i,j)=p_{s}(i,j+1)$. First we assume $n_{d}(i,j)=n_{d}(i,j+1)=0$. Inserting into Equation \eqref{equ:probability_remove_box_given_location} gives
    \begin{align}
        p_{s}(i,j)=\frac{n_{r}(i,j)-1}{n_{r}(i,j)}p_{s}(i,j+1)+\frac{1}{p_{s}(i,j)}p_{s}(i,j+1)=p_{s}(i,j+1) \, .
    \end{align}
    Now we assume that $p_{s}(k,j)=p_{s}(k,j+1)$ for all $i<k\leq i+n_{d}(i,j)$. Then we have
    \begin{align}
        p_{s}(i,j)
        &=\frac{1}{n_{r}(i,j)+n_{d}(i,j)}\left(\sum_{i'=i+1}^{i+n_{d}(i,j)}p_{s}(i',j)+\sum_{j'=j+1}^{j+n_{r}(i,j)}p_{s}(i,j')\right)= \\
        &=\frac{1}{n_{r}(i,j)+n_{d}(i,j)}\left(\sum_{i'=i+1}^{i+n_{d}(i,j)}p_{s}(i',j+1)+p_{s}(i,j+1)+\sum_{j'=j+2}^{j+n_{r}(i,j)}p_{s}(i,j')\right)=\\
        &=\frac{1}{n_{r}(i,j)+n_{d}(i,j)}\left((n_{r}(i,j)+n_{d}(i,j)-1)p_{s}(i,j+1)+p_{s}(i,j+1)\right)=p_{s}(i,j+1) \, .
    \end{align}
    In the same way, we show that $p_{s}(i,j)=p_{s}(i+1,j)$ if $n_{r}(i,j)=n_{r}(i+1,j)$.

    This means that if we have $n_{r}(i,j)=...=n_{d}(i+a,j)$ and $n_{d}(i,j)=...=n_{d}(i,j+b)$, then the boxes $(i',j')$ with $i\leq i' \leq i+a$ and $j\leq j' \leq j+b$ forming a rectangle of width $b$ and height $a$ are equivalent for the algorithm. We therefore contract this rectangle into a single box $(k,l)$, and we keep track of the number of original boxes by remembering $v(k)=a$ and $w(l)=b$. This process is shown in Equation \eqref{equ:squashing_boxes_YT}. Instead of keeping track of specific boxes as in Algorithm \ref{alg:sample_one_box_removed}, we now simply track the rectangles represented by the new boxes $(k,l)$. However, in the original algorithm, we chose all boxes uniformly at random. We remedy this by weighting the contracted rectangle according to their widths and heights. In the first step, we pick any box, so we weight the rectangles by $v(k)\cdot w(l)$. In the subsequent steps, we just pick a box from a given row or column, so we just weight the rectangles by $w(l)$ and $v(k)$ respectively. Additionally, changing to a box within the same rectangle won't change the final probabilities, so we just omit these boxes.
\end{proof}

\end{document}